\documentclass[12pt]{article}
\usepackage{geometry}
\usepackage[colorlinks=true, allcolors=blue]{hyperref}
\usepackage[round, authoryear]{natbib}
\usepackage{amsmath}
\usepackage{amssymb, amsfonts}
\usepackage{mathrsfs}
\usepackage{mathtools}
\usepackage{amsthm}
\usepackage{bm}
\usepackage{algorithm}
\usepackage{algorithmic}

\newtheorem{assumption}{Assumption}
\newtheorem{theorem}{Theorem}
\newtheorem{remark}{Remark}
\newtheorem{example}{Example}
\newtheorem{proposition}{Proposition}
\newtheorem{lemma}{Lemma}
\newtheorem{definition}{Definition}

\usepackage{threeparttable}
\usepackage{booktabs}
\usepackage{enumerate}
\usepackage{multirow}
\usepackage{subfigure}

\def\th{\widehat{\tau}}
\def\Kh{\widehat{K}}
\def\Pr{\mathbb{P}}
\def\D{\widehat{\mathcal{T}}}
\def\R{\widehat{\mathcal{R}}}
\def\T{\mathscr{T}}

\def\bbR{\mathbb{R}}
\def\bfx{{x}}
\def\bfS{\mathbf{S}}
\def\E{\mathbb{E}}
\def\bv{{v}}
\def\bomega{\omega}
\def\bbI{1}
\def\bZ{\kappa}
\def\e{e}
\def\Var{Cov}
\def\bfY{\bm{\bZ}}
\def\bSigma{\Sigma}

\def\Var{{\rm Var}}
\def\Cov{{\rm Cov}}
\def\tg{\widetilde{g}}
\def\bfepsilon{\bm{\epsilon}}

\def\bL{\mathbb{L}}
\def\bV{\mathbb{V}}
\usepackage{multirow}
\usepackage{makecell}
\usepackage{caption}

\def\S{S}

\begin{document}
\title{\bf\Large TUNE: Algorithm-Agnostic Inference after Changepoint Detection}
\author{Yinxu Jia, Jixuan Liu, Guanghui Wang, Zhaojun Wang, and Changliang Zou \\
{\small\it School of Statistics and Data Science, Nankai University, China} \\
}
\date{}
\maketitle
\baselineskip 20pt

\begin{abstract}
In multiple changepoint analysis, assessing the uncertainty of detected changepoints is crucial for enhancing detection reliability---a topic that has garnered significant attention. Despite advancements through selective p-values, current methodologies often rely on stringent assumptions tied to specific changepoint models and detection algorithms, potentially compromising the accuracy of post-detection statistical inference. We introduce TUNE (Thresholding Universally and Nullifying change Effect), a novel algorithm-agnostic approach that uniformly controls error probabilities across detected changepoints. TUNE sets a universal threshold for multiple test statistics, applicable across a wide range of algorithms, and directly controls the family-wise error rate without the need for selective p-values. Through extensive theoretical and numerical analyses, TUNE demonstrates versatility, robustness, and competitive power, offering a viable and reliable alternative for model-agnostic post-detection inference.
\end{abstract}

\vspace{0.2cm}

\section{Introduction}\label{sec:intro}

Changepoint analysis is essential for detecting significant changes in the parameters or distributions within data streams. The detection of multiple changepoints has become increasingly relevant with the growth of temporal data, serving as a pivotal element in modeling, estimation, and inference.

Research over the past two decades has primarily focused on developing diverse algorithms to localize changepoints across various models. These algorithms fall into two main categories: dynamic programming-based algorithms \citep{Killick+Rebecca+Fearnhead+etal-2012-p1590}, estimating all changepoints simultaneously, and greedy algorithms, identifying changepoints sequentially. The latter includes binary segmentation (BS) and its variants \citep{Fryzlewicz-2014-p2243,kovacs+Buhlmann-2023-p249}, as well as moving-window methods \citep{Niu+Zhang-2012-p1306,Eichinger+Kirch-2018-p526}. For an extensive review, see \cite{truong-2020-p107299}. Recent developments have expanded the scope of model settings to encompass high-dimensional mean and covariance models \citep{bai2010common, jirak2015, cho2015multiple, wang+Samworth+2018+p57, wang2018change, dette2022estimating, liu2020unified, yu2021finite, zhang2022adaptive, wang+Feng-2023-p936}, high-dimensional linear models \citep{lee2016lasso, leonardi2016computationally, kaul2019efficient, wang2021statistically, MR4784067}, and nonparametric models \citep{MR3210993, lung2015homogeneity, MR3180567, MR4048973, chen2023graph}. Additionally, modern machine learning techniques are increasingly integrated into detection algorithms \citep{londschien2023random, li2024automatic}.

Recent literature emphasizes the consistency of point estimators, which often relies on stringent assumptions regarding specific models and detection algorithms. However, perfect model recovery is not always attainable in practice, resulting in potential inaccuracies or errors in subsequent statistical inference based on detected changepoints. This has spurred growing interest in evaluating the uncertainty of detected changepoints to enhance the reliability of detection outcomes. In this context, \cite{frick2014multiscale} pioneered the Simultaneous MUltiscale Changepoint Estimator (SMUCE), which utilizes multiscale likelihood ratio statistics to control the family-wise error rate (FWER), thus mitigating the risk of overestimating the true number of changepoints. Later developments include the work by \cite{pein2017heterogeneous} and \cite{jula2022multiscale}. Moreover, \cite{Li2016fdr} adapted SMUCE by incorporating local quantiles for controlling the false discovery rate (FDR). These SMUCE-based methods primarily focus on the number of changepoints rather than their precise locations. In addition, \cite{hao+Niu+Zhang-2013-p1553} employed moving sum (MOSUM) statistics alongside the Benjamini–Hochberg procedure to control FDR. Furthermore, \cite{MR4766015} introduced the Narrowest Significance Pursuit (NSP) approach to identify localized regions, each guaranteed to contain a changepoint, while maintaining a predefined global significance level; see also \cite{Fang+Li+Siegmund-2020-p1615}. However, these methods are often tailored to specific models and algorithms. Methods like sample-splitting or cross-validation, as explored by \cite{chen2023data}, \cite{wang2022data}, and \cite{chen2024uncertainty}, introduce flexibility but at the expense of statistical power. While these strategies provide a robust framework for uncertainty assessment across various models and detection algorithms, a significant concern persists: using the full sample for changepoint detection is generally preferable to avoid accuracy losses, given the discrepancies that often arise between results obtained from full and partial samples.

To validate the detection of a changepoint by an algorithm, a straightforward approach involves conducting a two-sample test, such as the classical t-test for mean changes, using samples from adjacent segments around the detected changepoint. However, this method typically yields invalid p-values due to the reuse of data for both detection and subsequent testing---a phenomenon recognized as \textit{``double dipping''} in the literature \citep{kriegeskorte2009circular}. Recent advancements have shifted towards the use of \textit{selective inference} \citep{Fithian+Sun+Taylor-2017}, as demonstrated in studies by \cite{Hyun+GSell+Tibshirani-2018-p1053}, \cite{Hyun+Lin+GSell+Tibshirani-2021-p1037}, \cite{Duy+Toda+Sugiyama+Takeuchi-2020-p11356}, \cite{Jewell+Fearnhead+Witten-2022-p1082}, and \cite{Carrington+Fearnhead-2023}. These studies introduce the concept of the \textit{selective p-value}, which is the p-value conditioned on the selection of the test itself. To make the conditional probability computationally feasible, these approaches often require conditioning on additional information. However, overconditioning can lead to a reduction in statistical power \citep{Jewell+Fearnhead+Witten-2022-p1082,Carrington+Fearnhead-2023}. Selective methods are generally designed for specific models, primarily mean change models for univariate normal data. Adapting these methods to different detection algorithms, such as fixed-step generalized lasso \citep{Hyun+GSell+Tibshirani-2018-p1053}, fixed-step BS \citep{Hyun+Lin+GSell+Tibshirani-2021-p1037}, and Pruned Exact Linear Time \citep[PELT,][]{Killick+Rebecca+Fearnhead+etal-2012-p1590} with a fixed penalty \citep{Jewell+Fearnhead+Witten-2022-p1082,Duy+Toda+Sugiyama+Takeuchi-2020-p11356}, requires substantial customization. Although extensions to other settings are possible, they necessitate conditioning on more information, which increases computational demands and may further reduce statistical power. For instance, \cite{Hyun+GSell+Tibshirani-2018-p1053} and \cite{Hyun+Lin+GSell+Tibshirani-2021-p1037} developed information criterion-based methods to adaptively select the number of steps, and \cite{Hyun+Lin+GSell+Tibshirani-2021-p1037} proposed a sampling strategy that conditions on a sufficient statistic of the error variance to tackle unknown variances. Additionally, \cite{Jewell+Fearnhead+Witten-2022-p1082} recommended using a robust error variance estimator when variances are unknown.

The review and discussion thus far highlight several unresolved gaps. Firstly, the adaptation of selective methods to a wide range of state-of-the-art detection algorithms such as segment neighborhood \citep{auger+Lawrence-1989-p39}, PELT, seeded BS \citep{kovacs+Buhlmann-2023-p249}, etc., remains a significant challenge. This challenge is particularly pronounced when various model selection criteria, like 
Bayesian information criteria \citep[BIC,][]{yao1988estimating} and cross-validation \citep{Zou+Wang+Li-2020-p413}, are employed to determine the number of changepoints. Secondly, designing new selective methods that are compatible with more general data structures or distributions proves intricate. This complexity raises questions about how to extend beyond the mean change model for univariate normal data, particularly given the complexities of modern datasets. Finally, addressing the issue of multiplicity in evaluating the uncertainty of multiple detected changepoints continues to be a critical concern. While some selective methods incorporate Bonferroni corrections \citep{Hyun+Lin+GSell+Tibshirani-2021-p1037}, they lack theoretical guarantees for controlling specific error rates, especially when the number of detected changepoints is determined in a data-driven manner rather than being pre-specified.

\subsection{Our Contributions}

This paper introduces TUNE (Thresholding Universally and Nullifying change Effect), a simple yet effective approach designed to address the limitations of selective methods in post-detection inference. TUNE consists of two main elements: the construction of a test statistic for each detected changepoint and the establishment of a threshold to assess its significance. The key idea is to overestimate the null distribution of each test statistic by exploring all possible changepoint locations, thereby achieving uniform control over error probabilities across all detected changepoints. There is potential concern about the over-conservativeness of TUNE, often arising from experiences in regression modeling where the model space can expand exponentially with the covariate dimension. However, these issues are less pronounced in the context of inference post changepoint detection, where the model space is inherently limited by the sample size. Both theoretical and numerical analyses validate that TUNE consistently controls the error rate while maintaining competitive power relative to existing methods.

The main contributions of this paper are outlined as follows:

\begin{itemize}
\item TUNE focuses on directly controlling the FWER, which remains well-defined despite the data-dependent nature of specifying null hypotheses.

\item The control of FWER is achieved through a universal threshold for multiple tests, which can be computed flexibly and efficiently, independent of the detection algorithm employed. This threshold can be determined by approximating the distribution of the maximum of a sequence of localized two-sample statistics under a hypothetical global ``null'' scenario---no changes in the parameter of interest would have occurred, which connects to MOSUM methods for changepoint testing and presents an area of independent interest. We offer new theoretical insights into this approximation for high-dimensional parametric changepoint models by employing multiplier bootstrap methods with difference-based centering techniques.

\item The choices of two-sample statistics can be highly flexible, allowing adaptation to scenarios involving high-dimensionality, nonparametric settings, or additional contextual information from the changepoint model.
\end{itemize}

\subsection{Notations}

For an integer $n>0$, let $[n]$ denote the set $\{1,\ldots,n\}$. The notation $a_n\asymp b_n$ implies the existence of positive constants $c$ and $C$ such that $c\le|a_n/b_n|\le C$. For a vector $x\in \bbR^d$, define the $L_p$-norm $\Vert x\Vert_p=(\sum_{i=1}^d \vert x_i\vert^p)^{1/p}$ for a positive integer $p$, and the $L_{\infty}$-norm $\Vert x\Vert_{\infty}=\max_{i=1, \dots, d}\vert x_i\vert$. {For a symmetric matrix $A\in\bbR^{d\times d}$, let $\|A\|_2=\sup_{x\in\bbR^d}\|Ax\|_2/\|x\|_2$ and $\|A\|_\infty=\max_{i, j=1,\dots, d}|A_{ij}|$.} The unit sphere in $\mathbb{R}^{d-1}$ is denoted by $\mathbb{S}^{d-1}=\{z\in \bbR^{d-1}:\Vert z\Vert_2=1\}$. For a sequence of data $\{Z_i:i\in[n]\}$ and an index set $I\subset(0,n]$, let $Z_I=\{Z_i:i\in I\}$ and define the sample mean of $Z_I$ as $\bar{Z}_I=|I|^{-1}\sum_{i\in I}Z_i$, where $|I|$ is the cardinality of $I$. The abbreviation ``iid'' stands for ``independent and identically distributed.'' The notation $N(\mu,\Sigma)$ denotes the normal distribution with mean $\mu$ and covariance matrix $\Sigma$.

\subsection{Structure}

The remainder of this paper is structured as follows. Section \ref{sec:framework} presents a general framework for post-detection inference, focusing on algorithm-agnostic FWER control. Section \ref{sec:example_and_extension} delves into a variety of changepoint models and data structures, illustrating the provable versatility of the proposed framework. Simulation studies and real-data analyses are detailed in Section \ref{sec:simul} and Section \ref{sec:realdata}, respectively. Section \ref{sec:conclusion} concludes the paper. All theoretical proofs are provided in the supplementary material.

\section{General Framework}\label{sec:framework}

\subsection{Post-Detection Inference via Multiple Testing}

Consider a sequence of independent observations $\{Z_i: i\in[n]\}$, each taking values from a set $\mathcal{Z}$ and has a distribution $P_i$. Of interest is to detect changes in a specific statistical property, or a \textit{parameter} $\theta_i\equiv\theta(P_i)$, such as the mean $\theta(P_i)=\int zdP_i$ or the distribution $\theta(P_i)=P_i$. We consider a general multiple changepoint model:
\begin{align}\label{MCP_model}
\theta_i=\theta^*_k,\ i\in(\tau^*_{k-1},\tau^*_k],\ k\in[K^*+1],
\end{align}
where the data is partitioned into $K^*+1$ segments by the changepoints $\{\tau^*_k: k\in[K^*]\}$, with the conventions $\tau^*_0=0$ and $\tau^*_{K^*+1}=n$. Observations within each segment are characterized by a common parameter $\theta^*_k$, where $\theta^*_{k-1}\ne\theta^*_k$ for $k\in[K^*]$. A changepoint detection algorithm identifies a set of potential changepoints, say
\[
\D=\{\th_1,\ldots,\th_{\Kh}\}.
\]
The algorithm may be tailored to detect specific types of changes, such as mean shifts, through global or local procedures, incorporating various model selection criteria to determine the number of changepoints. It may also leverage advanced machine learning techniques to accommodate complex data distributions. See Section \ref{sec:intro} for more details. Our objective is to evaluate the reliability of each detected changepoint, while controlling specific error rates.

To this end, we employ a multiple testing framework to examine the presence of any true changepoint within a local neighborhood of each detected changepoint:
\begin{align}\label{H0}
H_{0,\th_j}: \theta_{\th_j-h+1}=\cdots=\theta_{\th_j+h}\ \text{versus}\ H_{1,\th_j}:\ \text{not}\ H_{0,\th_j},\quad\text{for}\ j=1,\ldots,\Kh,
\end{align}
where $h>0$ is a predetermined window size. If $H_{0,\th_j}$ holds, suggesting no true changepoint within the neighborhood $(\th_j-h,\th_j+h]$, then $\th_j$ is deemed a \textit{true null}. The rejection of $H_{0,\th_j}$ implies a \textit{reliable} detection outcome. We define $\R = \{\th_j\in\D: H_{0,\th_j}\ \text{is rejected}\}$ as the set of detected changepoints confirmed as reliable.

\begin{example}[Univariate mean change model]\label{ex:uni_mean}
A fundamental instance of Model (\ref{MCP_model}) concerns the detection of mean changes in univariate data, where  $\mathcal{Z}=\mathbb{R}$ and $Z_i=\theta_i+\epsilon_i$ with iid noises $\epsilon_i$ from a distribution $P_\epsilon$ with mean zero. Recent advancements in post-detection inference have primarily focused on this model with $P_\epsilon=N(0,\sigma^2)$, assuming a known variance $\sigma^2$. This model has been explored through various specific detection algorithms \citep{Hyun+GSell+Tibshirani-2018-p1053, Hyun+Lin+GSell+Tibshirani-2021-p1037, Duy+Toda+Sugiyama+Takeuchi-2020-p11356, Jewell+Fearnhead+Witten-2022-p1082, Carrington+Fearnhead-2023}.
\end{example}

\begin{remark}[On specifying true nulls]
One might consider defining a true null as simply $\th_j$ is not a true changepoint, i.e., $\th_j\ne\tau^*_k$ for any $k\in[K^*+1]$. However, such a definition is overly stringent given the difficulties inherent in precisely locating a changepoint. For instance, in Example \ref{ex:uni_mean} with $K^*=1$, the best achievable precision is that $|\th_1-\tau^*_1|=O_P(1)$ \citep{csorgo1997limit}, indicating that $\th_1$ can deviate from $\tau^*_1$ by several points, yet still provide valuable information about changes. Alternative definitions of true nulls, such as no changes occurring between estimated changepoints immediately before and after $\th_j$, are discussed further in Section \ref{sec:H0_3}.
\end{remark}

The multiple testing framework (\ref{H0}) poses substantial challenges in defining meaningful individual error rates due to the \textit{data-dependent} nature of each null hypothesis. However, the FWER can be well-defined, independent of the detection algorithm used. Define $\mathcal{H}$ as the set of all potential changepoints that are true nulls:
\[
\mathcal{H}=\{\tau\in[n-1]:(\tau-h,\tau+h)\cap\{\tau^*_k\}_{k=1}^{K^*}=\emptyset\}.
\]
Thus, $H_{0,\th_j}$ is true if and only if $\th_j\in\mathcal{H}$. The FWER is then given by
\begin{align}\label{FWER}
{\rm FWER}\equiv\Pr(\text{there exists some $\th_j\in\R\cap\mathcal{H}$}),
\end{align}
where $\Pr$ denotes the probability taken with respect to all random quantities involved, encompassing the actual data distribution $\prod_{i=1}^nP_i$ and the inherent randomness of the detection algorithm utilized (e.g., random intervals in Wild BS \citep{Fryzlewicz-2014-p2243} and random data splitting in cross-validation-based model selection criteria \citep{Zou+Wang+Li-2020-p413}). This definition raises a critical question: can we directly control the FWER, particularly given the inherent challenges of multiplicity, regardless of the detection algorithm employed?

\subsection{Algorithm-Agnostic FWER Control}\label{sec:FWER}

For each null hypothesis $H_{0,\th_j}$, where $\th_j\in\D$, a testing procedure involves a test statistic and a threshold. To clarify, for any $\tau\in[n-1]$, we define the test statistic as
\[
T_\tau=\T(Z_{(\tau-h,\tau+h]}),
\]
where $\T$ is a generic operator applied to the observations within the local neighborhood $(\tau-h,\tau+h]$. This statistic typically represents a \textit{two-sample} test statistic that quantifies the discrepancies in the parameter of interest (cf. Model (\ref{MCP_model})) between the segments $Z_{(\tau-h,\tau]}$ and $Z_{(\tau,\tau+h]}$. For instance, in Example \ref{ex:uni_mean} where the noises are $N(0,\sigma^2)$ with a known $\sigma^2$, an appropriate test statistic is
\begin{align*}
T_{\tau,\rm mean}=\sqrt{h/2}\left\vert\bar{Z}_{(\tau-h,\tau]}-\bar{Z}_{(\tau,\tau+h]}\right\vert/\sigma,
\end{align*}
where a larger value indicates a rejection of $H_{0,\tau}$. Additional two-sample test statistics and alternative formulations of the operator $\T$ are discussed in Section \ref{sec:example_and_extension}.

Let $\alpha\in(0,1)$ be a prescribed significance level. Rather than deriving specific thresholds for each test statistic $T_{\th_j}$ to individually control Type-I error rates (which may be ambiguous in our context) or selective error rates, we advocate for identifying a subset of the detected changepoints considered reliable via
\[
\R=\{\th_j\in\D: T_{\th_j}>t_{\rm u,\alpha}\},
\]
where the \textit{universal threshold} $t_{\rm u,\alpha}$ is determined so that the FWER (\ref{FWER}) is controlled at $\alpha$.

The key idea is to overestimate the probability of any false rejection, by exploring all potential changepoint locations and transferring probability evaluations from localized nulls (cf. (\ref{H0})) to a hypothetical global ``null'' scenario. Specifically, we analyze the distribution of $\max_{\tau=h,\ldots,n-h}T_\tau$ under a probability $\Pr_0$, which restricts $\Pr$ to situations where no changes in the parameter of interest would have occurred across all data segments. Notably, the $\Pr_0$ considers a hypothetical scenario where $\theta^*_1=\cdots=\theta^*_{K^*+1}$, contrasting with $\Pr$ that reflects the dynamics of actual changepoints. The threshold is established by ensuring
\begin{align}\label{thresh}
\Pr_0(\max_{\tau=h,\ldots,n-h}T_\tau>t_{\rm u,\alpha})\le\alpha+o(1).
\end{align}
In the aforementioned example,
\[
\Pr_0(\max_{\tau=h,\ldots,n-h}T_{\tau,\rm mean}>t_{\rm u,\alpha})=\Pr(\max_{\tau=h,\ldots,n-h}\sqrt{h/2}\left\vert\bar{\xi}_{(\tau-h,\tau]}-\bar{\xi}_{(\tau,\tau+h]}\right\vert>t_{\rm u,\alpha}),
\]
where $\xi_i=\epsilon_i/\sigma$ are iid as $N(0,1)$. The threshold $t_{\rm u,\alpha}$ can be computed through simulations such that $ \Pr_0(\max_{\tau=h,\ldots,n-h}T_\tau>t_{\rm u,\alpha})\approx\alpha$. Generally, the establishment of the threshold entails approximating the distribution of a MOSUM-like statistic---derived by taking the maximum of a sequence of localized two-sample statistics---in scenarios without changepoints. Further discussion on how this threshold is determined to satisfy Eq. (\ref{thresh}) across various changepoint models and test statistics is elaborated in Section \ref{sec:example_and_extension}.

We designate our post-detection inference method as \textit{TUNE}, an acronym for \textit{Thresholding Universally and Nullifying change Effect}. This approach turns out to uniformly control the FWER, irrespective of the detection algorithm used, thus offering \textit{algorithm-agnostic} robustness. This capability depends on a critical property of the test statistics utilized, which we term ``\textit{nullifiability}''. Nullifiability serves to bridge the gap between $\Pr$ and $\Pr_0$.

\begin{assumption}[Nullifiability]\label{asmp:null}
$\Pr\big(\max_{\tau\in\mathcal{H}}T_\tau>t_{\rm u,\alpha}\big)
= \Pr_0\big(\max_{\tau\in\mathcal{H}}T_\tau>t_{\rm u,\alpha}\big) + o(1)$.
\end{assumption}

This assumption imposes a high-level requirement on the test statistics. In Example \ref{ex:uni_mean}, the statistic $T_{\tau,\rm mean}=\sqrt{h/2}\left\vert\bar{Z}_{(\tau-h,\tau]}-\bar{Z}_{(\tau,\tau+h]}\right\vert/\sigma$ readily fulfills this assumption because for any true null $\tau\in\mathcal{H}$, $T_{\tau,\rm mean}=\sqrt{h/2}\left\vert\bar{\epsilon}_{(\tau-h,\tau]}-\bar{\epsilon}_{(\tau,\tau+h]}\right\vert/\sigma$ is independent of any $\theta^*_k$. Consequently, $\Pr\big(\max_{\tau\in\mathcal{H}}T_{\tau,\rm mean}>t\big)=\Pr_0\big(\max_{\tau\in\mathcal{H}}T_{\tau,\rm mean}>t\big)$ for any $t$. This independence is intrinsic to two-sample statistics that compare sample means, cancelling out the common underlying population mean; refer to Sections \ref{sec:mean}--\ref{sec:score}. In general, asymptotic independence of the parameter of interest, such that the limiting distribution of $\max_{\tau\in\mathcal{H}}T_\tau$ is free of any $\theta^*_k$, is favorable and motivates a wide range of two-sample statistics that compare parameter estimators capable of admitting (asymptotic) linear expansions, with appropriate normalization; see Sections \ref{sec:Wald} and \ref{sec:nonpara}. Additional examples that conform to this assumption and further discussion of this principle are elaborated in Section \ref{sec:example_and_extension}.

\begin{theorem}\label{thm:FWER}
If the threshold is set according to Eq. (\ref{thresh}) and Assumption \ref{asmp:null} is satisfied, then FWER $\le\alpha+o(1)$.
\end{theorem}

The proof of Theorem \ref{thm:FWER} can be succinctly articulated as follows:
\begin{align*}
\begin{aligned}
\text{FWER} &= 
\Pr\big(T_{\th_j}>t_{\rm u,\alpha}\ \text{for some}\ \th_j\in\mathcal{H}\big)\\
&\stackrel{(i)}{\le} \Pr\big(\max_{\tau\in\mathcal{H}}T_\tau>t_{\rm u,\alpha}\big)\\
&\stackrel{(ii)}{=} \Pr_0\big(\max_{\tau\in\mathcal{H}}T_\tau>t_{\rm u,\alpha}\big)+o(1)
\stackrel{(iii)}{\le} \Pr_0\big(\max_{\tau=h,\ldots,n-h}T_\tau>t_{\rm u,\alpha}\big)+o(1)
\stackrel{(iv)}{\le} \alpha+o(1).
\end{aligned}
\end{align*}
Inequality (i) extends the focus from detected true nulls to all potential true nulls. Equality (ii) applies Assumption \ref{asmp:null}. Inequality (iii) widens the consideration to include all possible changepoint locations. Finally, Inequality (iv) is justified by the threshold determination strategy (cf. Eq. (\ref{thresh})). A key feature of this proof is its algorithm-agnostic nature, which is achieved by removing any dependency on detected changepoints through Inequality (i).

The scaling step (i) of the proof might prompt concerns regarding the potential conservativeness of the TUNE method. Such concerns often arise from experiences with simultaneous inference methods in regression scenarios, where the space of null models could expand exponentially with the covariate dimension $d$ (potentially as large as $2^d$) \citep{berk2013valid}. However, in the context of changepoint inference addressed here, this issue is much less severe. The number of all possible true nulls is generally limited to at most $n$, and taking a maximum over the corresponding statistics often yields a $\sqrt{\log n}$ inflation. Our numerical analyses further demonstrate that TUNE consistently controls the FWER while remaining competitive power compared to selective methods.

\begin{remark}[{On power}]
Consider a toy example: in Example \ref{ex:uni_mean}, we employ a MOSUM detection procedure. This procedure determines the presence of at least one changepoint by evaluating whether $\max_{h\le\tau\le n-h}T_{\tau,\rm mean}>t_{\rm det,\alpha}$, where the detection threshold $t_{\rm det,\alpha}$ is chosen such that $\Pr_0(\max_{h\le\tau\le n-h}T_{\tau,\rm mean}>t_{\rm det,\alpha})\to \alpha$. Under mild conditions, $t_{\rm det,\alpha}\asymp\sqrt{\log(n/h)}$, as supported by Theorem 2.1 in \cite{Eichinger+Kirch-2018-p526}. The consistency of this test is assured if $\min_{k\in[K^*]}|\theta^*_{k+1}-\theta^*_k|/\{h^{-1/2}(\log n)^{1/2}\}\to\infty$; see Theorem 2.2 in \cite{Eichinger+Kirch-2018-p526}. Upon a rejection by this test, the procedure identifies a changepoint at $\th_1=\arg\max_{h\le\tau\le n-h}T_{\tau,\rm mean}$. For the subsequent post-detection inference, employing the TUNE method with $T_\tau=T_{\tau,\rm mean}$ and, in particular, $t_{\rm u,\alpha}=t_{\rm det,\alpha}$ immediately ensures that ${\rm FWER}\to\alpha$, according to Eq. (\ref{thresh}) and Theorem \ref{thm:FWER}. In fact, this method can re-detect any non-null $\th_1$ during the post-detection phase at nearly the same rate as during the initial changepoint testing stage; refer to Section \ref{sec:max_CUSUM} for detailed discussions.
\end{remark}

\section{Examples and Extensions}\label{sec:example_and_extension}

\subsection{Two-Sample Test Statistics and Thresholds}\label{sec:two_sample}

Theorem \ref{thm:FWER} lays the groundwork for achieving algorithm-agnostic control of the FWER, prompting vital inquiries into its practical applications: what types of two-sample test statistics can be ``nullified'' to satisfy Assumption \ref{asmp:null}, and how should universal thresholds be determined according to Eq. (\ref{thresh})? This section delves into various changepoint models and data structures, providing practical methods for selecting suitable two-sample test statistics and thresholds. To facilitate theoretical discussions, we will assume that the number of true changepoints, $K^*$, is fixed.

\subsubsection{Warm-Up: Revisiting Example \ref{ex:uni_mean}}\label{sec:mean_uni_normal}

Example \ref{ex:uni_mean} with normal noises $N(0,\sigma^2)$ for a known $\sigma^2$ is thoroughly discussed in post-detection inference literature and serves as an instructive baseline. As demonstrated in Section \ref{sec:FWER}, the two-sample mean test statistic $T_\tau=T_{\tau,\rm mean}$ is employed, and the threshold $t_{\rm u,\alpha}$ is set as the upper-$\alpha$ quantile of the distribution of
\[
\max_{\tau=h,\ldots,n-h}\sqrt{h/2}\left\vert\bar{\xi}_{(\tau-h,\tau]}-\bar{\xi}_{(\tau,\tau+h]}\right\vert,
\]
which can be efficiently simulated. In fact, the TUNE method can be applied to Example \ref{ex:uni_mean} without the normality and known variance assumptions, as illustrated via Example \ref{ex:Wald} below, which includes Example \ref{ex:uni_mean} as a special instance.

\begin{remark}[On computational complexities]
Post-detection inference in this example involves computing $T_{\th_j,\rm mean}$ for each $j\in[\Kh]$, with a computational complexity of $O(\min\{\Kh h,n\})$. Our method, TUNE, differs from selective methods in how thresholds are determined. TUNE sets a universal threshold, either via simulation---with a complexity of $O(Bn)$, where $B$ is the number of replications and $n$ is the computational cost per replication---or through utilization of the asymptotic distribution (cf. Section \ref{sec:Wald}), achieving a constant complexity of $O(1)$. This complexity is consistent, regardless of the detection algorithm used. In contrast, selective methods' complexities may vary depending on the specific algorithms. For example, the $k$-step BS described in \cite{Jewell+Fearnhead+Witten-2022-p1082}, which detects $\Kh=k$ changepoints, necessitates defining a conditional set at each detected changepoint---a union $I_j$ of several intervals, each demanding a complexity of at least $O(nk)$ (to implement another $k$-step BS). The number of these intervals, $|I_j|$, is dynamically determined, reaching up to a maximum of $2^k k!$. Consequently, the overall computational demand of this process can escalate to $O(nk\sum_{j=1}^k|I_j|)$---accounting for multiplicity, which could become considerably intensive when numerous changepoints are detected.
\end{remark}

\subsubsection{Estimating-Equation-Induced Parametric Changepoint Models}\label{sec:Wald}

\begin{example}[Changes in parameters induced by estimating equations]\label{ex:Wald}
Consider a scenario where each parameter $\theta^*_k\in\mathbb{R}^d$ is determined as the unique solution to an estimating equation $\E\{\psi_\theta(Z_i)\}=0$ for $i\in(\tau^*_{k-1},\tau^*_k]$, where $\psi_\theta$ is typically a known vector-valued function with $d$ coordinates. This framework accommodates various parameter types such as means (with $\psi_\theta(Z_i)=Z_i-\theta$), medians, and regression coefficients in linear regression, among others. Of interest is to detect changes in these parameters.
\end{example}

Within a segment $Z_I$ for $I\subset(0,n]$, a natural parameter estimator is the M-estimator, defined as the solution, say $\widehat{\theta}_I$, to the estimating equation $\sum_{i\in I}\psi_\theta(Z_i)=0$. For $i\in(\tau^*_{k-1},\tau^*_k]$, denote $V_i=V^*_k=\E\{\psi^{\prime\top}_{\theta^*_k}(Z_i)\}$, $\Sigma_i=\Sigma^*_k=\Var\{\psi_{\theta^*_k}(Z_i)\}$, and $\Gamma_i=\Gamma^*_k=V_k^{*-1}\Sigma^*_k(V_k^{*-1})^\top$. We employ the two-sample Wald test statistic
\begin{align*}
T_{\tau,\rm Wald} = \sqrt{h/2}\|\widehat{\Gamma}_\tau^{-1/2}(\widehat{\theta}_{(\tau-h,\tau]}-\widehat{\theta}_{(\tau,\tau+h]})\|_2,
\end{align*}
where $\widehat{\Gamma}_\tau$ is a reasonable estimator of $\Gamma_\tau$. In the absence of any changepoints, classical M-estimation theory implies that $T_{\tau,\rm Wald}^2$ is asymptotically distributed as a chi-squared distribution with $d$ degrees of freedom, independent of the underlying model parameters. Inspecting the proof of Theorem \ref{thm:Wald}, under mild conditions, the limiting distribution of $\max_{\tau\in\mathcal{H}}T_{\tau,\rm Wald}$ remains pivotal, thus fulfilling Assumption \ref{asmp:null}. To determine the universal threshold, it is necessary to establish the limiting distribution of $\max_{\tau=h,\ldots,n-h}T_{\tau,\rm Wald}$ under the probability $\Pr_0$, restricted to $\theta^*_1=\cdots=\theta^*_{K^*+1}$. Interestingly, this task parallels that in changepoint \textit{testing} problems, namely approximating the distribution of MOSUM statistics under the null hypothesis of no changes \citep{Eichinger+Kirch-2018-p526,Kirch2024}. Despite different contexts, the MOSUM-like nature of this setup and our scaling and nullifying ideas achieve the same end. This limiting distribution has already been characterized in \cite{Kirch2024}. The following assumption is required. 

\begin{assumption}\label{asmp:Wald}
(i) (Moments) There exists a constant $\nu>0$ such that $\E\{\|\psi_{\theta^*_k}(Z_i)\|_2^{2+\nu}\}<\infty$ for all $i\in(\tau^*_{k-1},\tau^*_k]$ and $k\in[K^*+1]$.
(ii) (Linearity) $\max_{\tau^*_{k-1}<\tau\le\tau^*_k-h}\|-\sqrt{h}V^*_k(\widehat{\theta}_{(\tau,\tau+h]}-\theta^*_k)-h^{-1/2}\sum_{i\in(\tau,\tau+h]}\psi_{\theta^*_k}(Z_i)\|_2 = o_P(\{\log(n/h)\}^{-1/2})$ for all $k\in[K^*+1]$.
(iii) (Window size) $n/h\to\infty$ and $\{n^{2/(2+\delta)}\log n\}/h\to 0$ for some $0<\delta<\nu$. 
(iv) (Variance estimators) $\max_{\tau\in(\tau^*_{k-1}+h,\tau^*_k-h]}\|\widehat{\Gamma}_\tau-\Gamma^*_k\|_2=o_P(\{\log(n/h)\}^{-1})$ for all $k\in[K^*+1]$.
\end{assumption}

Assumption \ref{asmp:Wald} is adapted from \cite{Kirch2024}. The existence of moments in Assumption \ref{asmp:Wald}(i) supports a strong invariance principle \citep{Einmahl1987invariance}, allowing partial sums of $\psi_{\theta^*_k}$ to be uniformly approximated by those of a Wiener process. Assumption \ref{asmp:Wald}(ii) posits that these partial sums are the leading terms in asymptotic linear expansions of the M-estimators, which can be justified under classical regularity conditions (see, for example, Regularity condition 2 in \cite{Kirch2024}). Assumption \ref{asmp:Wald}(iii) manages the scaling of the window size $h$, and Assumption \ref{asmp:Wald}(iv) ensures consistent covariance estimators. 

Under Assumption \ref{asmp:Wald} and $\Pr_0$, $a(n/h)\max_{\tau=h,\ldots,n-h}T_{\tau,\rm Wald}-b(n/h)$ converges to a Gumbel extreme value distribution $G$ in distribution, i.e., $\Pr(G\le t)=\exp\{-2\exp(-t)\}$, where $a(x)=\sqrt{2\log x}$, $b(x)=2\log x+(d/2)\log\log x-\log\big((2/3)\Gamma(d/2)\big)$, and $\Gamma$ is the gamma function.

\begin{theorem}\label{thm:Wald}
Suppose Assumption \ref{asmp:Wald} holds. The TUNE method, using the test statistic $T_{\tau,\rm Wald}$ and setting the threshold $t_{\rm u,\alpha}=\{b(n/h)+G^{-1}(1-\alpha)\}/a(n/h)$, satisfies FWER $\le\alpha+o(1)$.
\end{theorem}

Theorem \ref{thm:Wald} underscores that our proposed Wald-type inference scheme achieves asymptotically valid FWER control within the context of Example \ref{ex:Wald}. This method relies on the consistent and change-robust estimation of the covariance matrices $\Gamma^*_k$, which can be challenging in practice, particularly in high-dimensional settings. In Sections \ref{sec:mean}--\ref{sec:score}, we explore alternative \textit{score}-based strategies and discuss the application of multiplier bootstrap techniques, which are particularly beneficial in high-dimensional changepoint models.

\subsubsection{High-Dimensional Mean Changepoint Models}\label{sec:mean}

\begin{example}[Mean changes in high-dimensional data]\label{ex:mean}
Consider $\mathcal{Z}=\mathbb{R}^d$ and $Z_i=\theta_i+\epsilon_i$ with iid noises $\epsilon_i$ from a distribution $P_\epsilon$ with mean zero, which extends Example \ref{ex:uni_mean} to multidimensional scenarios. In recent years, there has been a growing interest in detecting mean changes in high-dimensional settings where $d\to\infty$. Various methods for aggregating componentwise evidence of mean changes have been explored, including the use of the $\ell_2$-norm \citep{bai2010common}, $\ell_\infty$-norm \citep{jirak2015}, and adaptive strategies \citep{cho2015multiple,liu2020unified,zhang2022adaptive,wang+Feng-2023-p936}.
\end{example}

Consider the two-sample mean test statistic
\begin{align*}
T_{\tau,{\rm mean},g}=g\big(\sqrt{h/2}(\bar{Z}_{(\tau-h,\tau]}-\bar{Z}_{(\tau,\tau+h]})\big),
\end{align*}
where $g\ge 0$ acts as an aggregation function that combines componentwise mean differences, such as the $\ell_\infty$-norm. This statistic is inherently independent of the underlying mean $\theta^*_k$ for a true null $\tau$, ensuring that 
\[
\Pr\big(\max_{\tau\in\mathcal{H}}T_{\tau,{\rm mean},g}>t\big)
= \Pr_0\big(\max_{\tau\in\mathcal{H}}T_{\tau,{\rm mean},g}>t\big)
= \Pr\big(\max_{\tau\in\mathcal{H}}g\big(\sqrt{h/2}(\bar{\epsilon}_{(\tau-h,\tau]}-\bar{\epsilon}_{(\tau,\tau+h]})\big)>t\big),
\]
for any $t$, thus satisfying Assumption \ref{asmp:null}.

To set the threshold $t_{\rm u,\alpha}$ as defined in Eq. (\ref{thresh}), it is essential to approximate the distribution of $\max_{h\leq\tau\leq n-h}T_{\tau,{\rm mean},g}$ under $\Pr_0$, that is, the distribution of
\begin{align}\label{ts_null:mean_g}
\max_{h\leq\tau\leq n-h}g\big(\sqrt{h/2}(\bar{\epsilon}_{(\tau-h,\tau]}-\bar{\epsilon}_{(\tau,\tau+h]})\big).
\end{align}
We propose employing a \textit{multiplier bootstrap} approach assisted by \textit{difference}-based strategies to facilitate this. The bootstrap statistic is defined as:
\begin{align*}
T_{\tau,\rm boots} = g\Big(\sqrt{h/2}\big(h^{-1}\sum_{i=\tau-h+1}^\tau \e_i(Z_{i+1}-Z_i)/\sqrt{2}-h^{-1}\sum_{i=\tau+1}^{\tau+h}\e_i(Z_i-Z_{i-1})/\sqrt{2}\big)\Big),
\end{align*}
with iid $N(0,1)$ multipliers $\e_i$. The threshold is then identified as the conditional upper $\alpha$-quantile of $\max_{h\leq\tau\leq n-h}T_{\tau,\rm boots}$ given the data, i.e.,
\begin{align}\label{thresh_boots}
t_{\rm u,\alpha} = \inf\{t: \Pr(\max_{h\leq\tau\leq n-h}T_{\tau,\rm boots}>t\mid \{Z_i\}_{i=1}^n)\le\alpha\},
\end{align}
which can be efficiently simulated in practice.

\begin{remark}[On multiplier bootstraps]
Investigating the limiting distribution of (\ref{ts_null:mean_g}) is of independent interest in changepoint testing scenarios. The choice of $g$ may necessitate specific dependence structures among components of $\epsilon_i$. This distribution could involve nuisance parameters hinging on high-dimensional noise distribution $P_\epsilon$, which presents challenges for their consistent estimation in the presence of multiple changepoints \citep{MR4528482}.

We engage high-dimensional bootstrap techniques \citep{chernozhukov+2023+p427}, which are increasingly utilized in high-dimensional two-sample testing and changepoint testing problems \citep{Xue2020distribution,yu2021finite,liu2020unified}. However, these methods are not directly applicable here. In two-sample testing, the distribution of $g(\bar{\epsilon}_{(\tau-h,\tau]}-\bar{\epsilon}_{(\tau,\tau+h]})$ for a given $\tau$ can be approximated by centering each dataset around its sample mean, regardless of the potential heterogeneity between the datasets $Z_{(\tau-h,\tau]}$ and $Z_{(\tau,\tau+h]}$ \citep{Xue2020distribution}. Changepoint testing introduces additional complexity due to unknown boundaries created by multiple changepoints. In cases of a single changepoint, the objective is to approximate the distribution of $\max_{h\le\tau\le n-h}g(\bar{\epsilon}_{(0,\tau]}-\bar{\epsilon}_{(\tau,n]})$. \cite{yu2021finite} and \cite{liu2020unified} recommended centering the data by sample means across different potential changepoint locations, using $Z_i-\bar{Z}_I$ for $i\in I$ where $I=(0,\tau]$ or $(\tau,n]$. This method successfully mimics the target distribution under the null hypothesis of no change, but may fall short under the presence of a changepoint. In our context, the challenge is to approximate the distribution of (\ref{ts_null:mean_g}) irrespective of the presence of multiple changepoints.
\end{remark}

\begin{assumption}\label{asmp:boots}
(i) (Aggregation function) For any $t\ge 0$, the set $\{x:g(x)\le t\}$ is $s$-sparsely convex for some integer $s>0$ (cf. Definition \ref{def:s-sparsity} in Section \ref{apdx:sec:details}). 
(ii) (Noises) For all $i\in[n]$ and any $\bv\in\mathbb{S}^{d-1}$ with $\lVert\bv\rVert_0\leq s$, $\E\{(\bv^\top\epsilon_i)^2\}\geq b$ for a constant $b>0$. There exists a sequence of constants $B_n\geq 1$, possibly diverging, such that for all $i\in[n]$ and $j\in[d]$, $\E(\lvert\epsilon_{ij}\rvert^{2+k})\leq B_n^k$ for $k=1,2$, and $\E\{\exp(\lvert\epsilon_{ij}\rvert/B_n)\}\leq 2$.
\end{assumption}

Assumption \ref{asmp:boots}(i) requires that all sublevel sets of the aggregation function $g$ should be $s$-sparsely convex, which holds for $g(x)=\|x\|_\infty$ with $s=1$ and $g(x)=\|x\|_2$ with $s=d$. Assumption \ref{asmp:boots}(ii) imposes specific constraints on the noise characteristics: noises within the $s$-sparsely convex set have variances bounded away from zero, the third and fourth moments of noise components do not increase too rapidly, and these components exhibit sub-exponential tails. These conditions mirror those often encountered in the high-dimensional bootstrap literature \citep{chern2017central}. 

\begin{theorem}\label{thm:mean}
Suppose that Assumption \ref{asmp:boots} holds with a fixed $s$. The TUNE method, using the test statistic $T_{\tau,{\rm mean},g}$ and setting the threshold $t_{\rm u,\alpha}$ as in Eq. (\ref{thresh_boots}), satisfies
\[
{\rm FWER} \leq \alpha + (nhd)^{-c} + C\{(B_n^2+\Delta_\theta^2)\log^7(nhd)/h\}^{1/6},
\]
for some positive constants $c$ and $C$, where $\Delta_\theta=\max_{k\in[K^*]}\|\theta^*_{k+1}-\theta^*_k\|_\infty^2$.
\end{theorem}

This theorem confirms that our TUNE method achieves asymptotically valid control of the FWER, independent of the detection algorithm used. It is proved by examining
\[
\sup_{t\geq 0}\Big\vert\Pr_0\big(\max_{h\leq\tau\leq n-h}T_{\tau,{\rm mean},g}\le t\big)-\Pr\big(\max_{h\leq\tau\leq n-h}T_{\tau,\rm boots}\le t\mid \{Z_i\}_{i=1}^n\big)\Big\vert,
\]
which characterizes the approximation accuracy of the multiplier bootstrap strategy to the target distribution, regardless of the presence of multiple changepoints. The proof  hinges on the fact that all sublevel sets of the composition function $\max\circ g$ remain sparsely convex, and thus it suffices to analyze a weighted sum of vectorized high-dimensional random vectors hitting these sets. Therefore, high-dimensional central limit theorems can be applied \citep{chern2017central}. It is then essential to demonstrate that difference-based strategies provide consistent covariance matrix estimation over these sparsely convex sets. In particular, Theorem \ref{thm:mean} is applicable for a fixed dimension $d$ with $g(x)=\|x\|_2$, offering a robust alternative to the Wald-type diagnostics discussed in Section \ref{sec:Wald}.

\subsubsection{Score-Induced Parametric Changepoint Models}\label{sec:score}

\begin{example}[Changes in parameters induced by scores]\label{ex:score}
Consider scenarios where \textit{scores}, $\S_i\equiv\S(Z_i)$, reflect changes in the parameters $\theta_i$ through their expectations. Specifically, if $\theta^*_k\ne\theta^*_{k+1}$, then $\mu^*_k\ne\mu^*_{k+1}$, where $\mu^*_k=\E(\S_i)$ for $i\in(\tau^*_{k-1},\tau^*_k]$. Denote $\Sigma^*_k=\Var(\S_i)$ for $i\in(\tau^*_{k-1},\tau^*_k]$, which may remain constant across segments ($\Sigma^*_k=\Sigma^*_{k+1}$) or vary ($\Sigma^*_k\ne\Sigma^*_{k+1}$). This framework includes Example \ref{ex:mean} by selecting $\S_i=Z_i$, where $\Sigma^*_k=\Sigma^*_{k+1}$. Alternatively, consider $Z_i=(y_i,X_i)\in\mathcal{Z}=\mathbb{R}\times\mathbb{R}^d$, where $y_i=X_i^\top\theta^*_k+\epsilon_i$, with covariates $X_i$ being iid with mean $0$ and covariance matrix $\Sigma_X$, and noises $\epsilon_i$ being iid with mean $0$. Choosing $\S_i=y_iX_i$ leads to $\mu^*_k-\mu^*_{k+1}=\Sigma_X(\theta^*_k-\theta^*_{k+1})$, where $\Sigma^*_k\ne\Sigma^*_{k+1}$ since $\Sigma^*_k$ depends on $\theta^*_k$.
\end{example}

Given the motivation from $T_{\tau,{\rm mean},g}$, consider the test statistic
\[
T_{\tau,\rm score}=g\big(\sqrt{h/2}(\bar{\S}_{(\tau-h,\tau]}-\bar{\S}_{(\tau,\tau+h]})\big),
\]
which, however, may not satisfy Assumption \ref{asmp:null} due to potential variations in $\Sigma^*_k$. A normalization similar to $T_{\tau,\rm Wald}$ might address this issue by providing asymptotic pivotalness, as demonstrated in Theorem \ref{thm:Wald}(i). However, this approach faces challenges in high-dimensional settings. {To circumvent this problem, we propose modifying $\Pr_0$ in Assumption \ref{asmp:null} to $\Pr^*_0$, which focuses solely on nullifying the means of the scores, i.e., $\mu^*_1=\cdots=\mu^*_{K^*+1}$. Consequently, it holds that $\Pr\big(\max_{\tau\in\mathcal{H}}T_{\tau,\rm score}>t\big) = \Pr^*_0\big(\max_{\tau\in\mathcal{H}}T_{\tau,\rm score}>t\big)$ for all $t$. By employing arguments from the proof of Theorem \ref{thm:FWER}, we establish the threshold $t_{\rm u,\alpha}$ such that
\[
\Pr^*_0\Big(\max_{h\leq\tau\leq n-h}g\big(\sqrt{h/2}(\bar{\S}_{(\tau-h,\tau]}-\bar{\S}_{(\tau,\tau+h]})\big)>t_{\rm u,\alpha}\Big)\to\alpha.
\]
Theorem \ref{thm:score} demonstrates that by applying the multiplier bootstrap method proposed in Section \ref{sec:mean} with $Z_i$ replaced by $S_i$, we can control the FWER asymptotically.}

\begin{theorem}\label{thm:score}
Suppose that Assumption \ref{asmp:boots} holds with $\epsilon_i=\S_i-\E(\S_i)$ and a fixed $s$. The TUNE method, using the test statistic $T_{\tau,\rm score}$ and setting the threshold $t_{\rm{u},\alpha}$ as in Eq. (\ref{thresh_boots}) but with $Z_i$ in $T_{\tau,\rm boots}$ replaced by $S_i$, satisfies
\[
{\rm FWER} \leq \alpha + (nhd)^{-c} +C\big(\{(B_n^2+\Delta_\mu^2)\log^7(nhd)/h\}^{1/6}+\{\Delta_\Sigma\log^2(nhd)/h\}^{1/3}\big),
\]
for some positive constants $c$ and $C$, where $\Delta_\mu=\max_{k\in[K^*]}\|\mu^*_{k+1}-\mu^*_k\|_\infty^2$ and $\Delta_\Sigma=\max_{k\in[K^*]}\|\Sigma^*_{k+1}-\Sigma^*_k\|_\infty$.
\end{theorem}

\subsubsection{Nonparametric Changepoint Models}\label{sec:nonpara}

\begin{example}[Changes in distributions]
Consider detecting distributional changes where $\theta_i=P_i$. Over recent years, the development of nonparametric changepoint detection methods has significantly expanded, such as those based on ranks \citep{MR3210993,lung2015homogeneity}, graphs \citep{chen2023graph}, energy distances \citep{MR3180567}, kernels \citep{MR4048973}, and random forests \citep{londschien2023random}, offering robust alternatives to traditional parametric methods.
\end{example}

In univariate scenarios where $\mathcal{Z}=\mathbb{R}$, we employ the two-sample Wilcoxon/Mann-Whitney test statistic:
\begin{align*}
T_{\tau,\text{WMW}}=\sum_{i=\tau+1}^{\tau+h}R_{\tau,i},
\end{align*}
where $R_{\tau,i}=\sum_{j=\tau-h+1}^{\tau+h}\bbI_{\{Z_j\le Z_i\}}$ is the rank of $Z_i$ among $Z_{(\tau-h,\tau+h]}$. The key observation is that the distribution of $\max_{\tau\in(\tau^*_{k-1}+h,\tau^*_k-h]}T_{\tau,\text{WMW}}$ relies only on the ranks of $Z_i$ within $(\tau^*_{k-1},\tau^*_k]$, free of the specific distribution $P^*_k$. This distribution-freeness property ensures that $\Pr\big(\max_{\tau\in\mathcal{H}}T_\tau>t\big) = \Pr_0\big(\max_{\tau\in\mathcal{H}}T_\tau>t\big)$ for any $t$, thereby satisfying Assumption \ref{asmp:null}. To establish the universal threshold $t_{\rm u,\alpha}$, it is noted that $\max_{h\leq \tau\leq n-h}T_{\tau,\rm WMW}$ is also distribution-free under $\Pr_0$ due to the rank-based nature. By simulating analogues---say $T_{\tau,\rm simul}$---of $T_{\tau,\rm WMW}$ with $\{Z_i:i\in[n]\}$ replaced by iid $N(0,1)$ random variables, the threshold can be set as
\begin{align}\label{thresh_simul}
t_{\rm u,\alpha} = \inf\{t:\Pr(\max_{h\le\tau\le n-h}T_{\tau,\rm simul}>t)\le\alpha\},
\end{align}
achievable through simulations.

\begin{theorem}\label{thm:rank}
If $Z_i\in\mathbb{R}$ has a continuous distribution, then the TUNE method, using the test statistic $T_{\tau,\rm{WMW}}$ and setting the threshold $t_{\rm u,\alpha}$ as in Eq. (\ref{thresh_simul}), satisfies $\text{FWER}\le\alpha$.
\end{theorem}

In multivariate settings where $\mathcal{Z}=\mathbb{R}^d$, adapting the TUNE method is interesting yet challenging. One effective approach is to use a test statistic based on componentwise ranks \citep{lung2015homogeneity}:
\[
T_{\tau,{\rm compr}} = \|\sqrt{2}h^{-3/2}\widehat{\Sigma}_{\tau}^{-1/2}\sum_{i=\tau+1}^{\tau+h}\{R_{\tau,i}-(2h+1)/2\}\|_2,
\]
where $\widehat{\Sigma}_\tau = 2h^{-1}\sum_{i=\tau-h+1}^{\tau+h}\{\widehat{F}_\tau(Z_i)-1/2\}\{\widehat{F}_\tau(Z_i)-1/2\}^\top$ and $\widehat{F}_\tau(t)=(2h)^{-1}\sum_{j=\tau-h+1}^{\tau+h}\bbI_{\{Z_j\leq t\}}$.

\begin{theorem}\label{thm:compr}
Let $Z_i\in\mathbb{R}^d$ have a continuous distribution function $F_i$ with $\Var\{F_i(Z_i)\}$ positive definite and possessing a finite largest eigenvalue. Assume that $n/h\to\infty$ and $\{n^{2/(2+\nu)}\log^3n\}/h\to 0$ for some $\nu>0$. The TUNE method, using the statistic $T_{\tau,{\rm compr}}$ and setting the threshold $t_{\rm u,\alpha}=\inf\{t:\Pr(\sup_{1\leq s<\infty}2^{-1/2}\|\{W(s)-W(s-1)\}-\{W(s+1)-W(s)\}\|_2>t)\le\alpha\}$, satisfies $\text{FWER}\leq \alpha+o(1)$, where $W$ denotes a $d$-dimensional standard Wiener process.
\end{theorem}

The proof of Theorem \ref{thm:compr} relies on uniform linear expansions of U-statistics, and application of invariance principles for these linear terms (cf. Theorem \ref{thm:Wald}). Under $\Pr_0$, $\max_{h \leq \tau \leq n-h} T_{\tau,\rm compr}$ is asymptotically identically distributed as $\sup_{1 \leq t < \infty} 2^{-1/2}\|\{W(t)-W(t-1)\}-\{W(t+1)-W(t)\}\|_2$, and thus is asymptotically distribution-free.

Alternatively, one may consider distance- or kernel-based test statistics, such as the energy distance-based statistic \citep{MR3180567}
\[
T_{\tau,\text{dist}}
=\frac{2}{h^2}\sum_{i=\tau-h+1}^\tau\sum_{j=\tau+1}^{\tau+h}d(Z_i,Z_j)
-\frac{1}{h(h-1)}\sum_{\tau-h<i\ne j\le\tau}d(Z_i,Z_j)
-\frac{1}{h(h-1)}\sum_{\tau<i\ne j\le\tau+h}d(Z_i,Z_j),
\]
where $d(Z_i,Z_j)=\|Z_i-Z_j\|_2$. The asymptotic distribution of $\max_{h\le\tau\le n-h}T_{\tau,\text{dist}}$ under $\Pr_0$ typically hinges on the unknown common underlying distribution, which complicates the setting of the universal threshold $t_{\rm u,\alpha}$. To address this issue, the permutation method proposed by \cite{MR3180567} may be used. This method involves recomputing $T_{\tau,\rm dist}$ with $\{Z_i\}$ replaced by $\{Z_{\pi_i}\}$ for a permutation $\pi=(\pi_1,\ldots,\pi_n)$ of the indices $[n]$, denoted as $T_{\tau,\rm perm}$. The threshold is then set as $t_{\rm u,\alpha}=\inf\{t:\Pr(\max_{h\le\tau\le n-h}T_{\tau,\rm perm}>t)\le\alpha\}$, which can be reliably estimated through simulation. Our simulation studies suggest that distance-based methods demonstrate robust performance in multivariate scenarios. However, a crucial aspect, particularly the verification of Assumption \ref{asmp:null}, remains a challenge, which necessitates further research.

\subsection{Other Strategies}

\subsubsection{Maximum of CUSUM Statistics}\label{sec:max_CUSUM}

Employing the maximum over a sequence of two-sample (or cumulative sum (CUSUM)) statistics offers an alternative to the single two-sample test statistic $T_\tau$, naturally connecting to changepoint testing literature. This approach involves comparing samples $Z_{(\tau-h,\ell]}$ and $Z_{(\ell,\tau+h]}$ across $\ell\in(\tau-h,\tau+h)$. We will illustrate this idea via Example \ref{ex:uni_mean} with $\sigma=1$.

Consider the maximum statistic
\begin{align*}
M_{\tau,\rm mean}=\max_{\tau-h+\lambda<\ell<\tau+h-\lambda}\sqrt{\frac{\{\ell-(\tau-h)\}(\tau+h-\ell)}{2h}}\left\vert\bar{Z}_{(\tau-h,\ell]}-\bar{Z}_{(\ell,\tau+h]}\right\vert,
\end{align*}
where $\lambda\in(0,h)$ is a boundary parameter. This statistic inherently satisfies Assumption \ref{asmp:null} as it cancels out the underlying mean parameter for a true null $\tau$. Following procedures in Sections \ref{sec:mean_uni_normal}--\ref{sec:Wald}, the threshold $t_{\rm u,\alpha}$ can be set as the upper-$\alpha$ quantile of the distribution of
\[
M_\epsilon=\max_{\tau=h,\ldots,n-h}\max_{\tau-h+\lambda<\ell<\tau+h-\lambda}\sqrt{\frac{\{\ell-(\tau-h)\}(\tau+h-\ell)}{2h}}\left\vert\bar{\epsilon}_{(\tau-h,\ell]}-\bar{\epsilon}_{(\ell,\tau+h]}\right\vert,
\]
which can be simulated due to the nature of pivotalness under proper normalization.

Utilizing $M_{\tau,\rm mean}$ can potentially enhance the power for post-detection inference compared to a single two-sample statistic by better capturing changes in the mean within a given neighborhood. Consider a class of alternatives where a single true changepoint exists within some neighborhood $(\tau-h,\tau+h]$ of a given $\tau\in[n-1]$:
\begin{align*}
\mathcal{H}_{1,\tau}(m_n)=\{H_{1,\tau}&:\theta_{\tau-h+1}=\cdots=\theta_{\tau^*_k}\ne\theta_{\tau^*_k+1}=\cdots=\theta_{\tau+h}\ \text{for some}\ k\in[K^*],\\
&\text{with}\ \{\tau^*_k-(\tau-h)\}\wedge\{(\tau+h)-\tau^*_k\}\asymp m_n\},
\end{align*}
where $m_n$ specifies the minimal sample size of the two heterogeneous segments $Z_{(\tau-h,\tau^*_k]}$ and $Z_{(\tau^*_k,\tau+h]}$. Define the separation rate
\[
\rho_n=\sqrt{\frac{m_n(2h-m_n)}{2h}}\left\vert\theta^*_{k-1}-\theta^*_k\right\vert.
\]
An ideal test against $H_{1,\tau}\in\mathcal{H}_{1,\tau}(m_n)$ would be a two-sample z-test comparing samples $Z_{(\tau-h,\tau^*_k]}$ and $Z_{(\tau^*_k,\tau+h]}$, as if $\tau^*_k$ were known. Classical large-sample theory suggests that this test is consistent when $\rho_n\to\infty$. In contrast, a two-sample z-test comparing samples $Z_{(\tau-h,\tau]}$ and $Z_{(\tau,\tau+h]}$ requires $m_n/\sqrt{2h}\left\vert\theta^*_{k-1}-\theta^*_k\right\vert\to\infty$, equivalently, $\sqrt{m_n/(2h-m_n)}\rho_n\to\infty$, to achieve consistency. If the samples are unbalanced (i.e., $m_n/(2h-m_n)\to 0$), a substantially larger signal is required for consistency compared to the ideal scenario. Proposition \ref{prop:power} demonstrates that our TUNE procedure, utilizing the maximum statistic $M_{\tau,\rm mean}$ and threshold $t_{\rm u,\alpha}$, achieves consistency if $\rho_n/\sqrt{\log(n/h)}\to\infty$, matching the ideal two-sample rate up to a logarithm factor. This approach effectively re-detects the maximal departure in mean change within the neighborhood $(\tau-h,\tau+h]$, especially beneficial when the originally detected changepoint is offset from the true changepoint. Moreover, this rate aligns with the minimax detection rate in classical changepoint testing contexts, adjusted for a logarithm factor \citep{csorgo1997limit}.

\begin{assumption}\label{asmp:power}
(i) (Moments) There exists a constant $\nu>0$ such that $\E(|\epsilon|^{2+\nu})<\infty$. (ii) (Window size) $n/h\to\infty$ and $\{n^{2/(2+\delta)}\log n\}/h\to 0$ for $0<\delta<\nu$. (iii) (Boundary) $\lambda\asymp h$.
\end{assumption}

\begin{proposition}\label{prop:power}
For any $\tau$ such that $H_{1,\tau}\in\mathcal{H}_{1,\tau}(m_n)$ for some $m_n>0$, $\Pr(M_{\tau,\rm mean}>t_{\rm u,\alpha})\to 1$ provided that $\rho_n/\sqrt{\log(n/h)}\to\infty$.
\end{proposition}

For any detected changepoint $\th_j\in\D$, let $\widehat{\Delta}_{jk}=|\th_j-\tau^*_k|$. As a consequence of Proposition \ref{prop:power}, if $\Pr(H_{1,\th_j}\in\mathcal{H}_{1,\th_j}(m_n))\to 1$, i.e., $\Pr(\widehat{\Delta}_{jk}\asymp h-m_n)\to 1$, then $\Pr(M_{\th_j,\rm mean}>t_{\rm u,\alpha})\to 1$.

\subsubsection{Self-Normalization Techniques}\label{sec:self-normalization}

In the examples explored thus far, our primary focus has been on independent datasets. This underlines the necessity for further research to adapt and extend the Wald- and score-based methods proposed in Sections \ref{sec:two_sample} to time series data. Such extensions could be realized by adapting techniques from \cite{Eichinger+Kirch-2018-p526} and \cite{Kirch2024} within the context of Example \ref{ex:Wald}. Additionally, extending normal approximation and bootstrap techniques for high-dimensional data \citep{MR3718156} might also be feasible for Example \ref{ex:score} by employing properly selected centering terms to adjust for change effects. These methods generally requires consistent and change-robust estimation of long-run (co)variances, which can be challenging.
Another promising approach involves leveraging recent advances in self-normalization techniques, which avoids the complexities of long-run covariance estimation while accommodating temporal dependence; see \cite{shao2015sn} for a comprehensive review. We illustrate how this technique can be incorporated into our proposed framework in Example \ref{ex:uni_mean}, providing a practical example of its application in a time series setting.

For $0\leq a< \ell\leq b\leq n$, define $L_{\ell, a, b}=(\ell-a)(b-\ell)/(b-a)^{3/2}(\bar{Z}_{(a, \ell]}-\bar{Z}_{(\ell, b]})$. Consider the locally self-normalized two-sample statistic
\begin{align*}
S_{\tau,\rm mean}=\frac{L_{\tau, \tau-h, \tau+h}^2}{h^{-1}(\sum_{j=\tau-h+1}^\tau L_{j, \tau-h, \tau}^2+\sum_{j=\tau+1}^{\tau+h} L_{j, \tau, \tau+h}^2)},
\end{align*}
which satisfies Assumption \ref{asmp:null}. According to Proposition \ref{prop:sn}, the limiting distribution of $\max_{h\leq\tau\leq n-h}S_{\tau,\rm mean}$ under $\Pr_0$ is pivotal, facilitating the determination of the threshold $t_{\rm u,\alpha}$ through simulation. Similar locally self-normalized statistics for testing the presence of a changepoint have been explored by \cite{shao2010testing}, \cite{MR4515555}, and \cite{cheng2024general_SN}. Proposition \ref{prop:sn}(i) adapts from Theorem 3.1 in \cite{cheng2024general_SN}.

\begin{proposition}\label{prop:sn}
Suppose that there exists a constant $\nu>0$ such that $\sum_{i=1}^n\epsilon_i-\sigma W(n)=O(n^{1/(2+\nu)})$ almost surely, where $\sigma^2=\lim_{n\to\infty}\Var(\sum_{i=1}^n\epsilon_i)/n\in(0,\infty)$ is the long-run variance, and $W$ is the standard Wiener process. If $n/h\to\infty$ and $\{n^{2/(2+\nu)}\log n\}/h\to0$, then: (i) Under $\Pr_0$, $\max_{h\leq\tau\leq n-h}S_{\tau,\rm mean}$ converges to $\sup_{1\leq t<\infty}{\bL_{t, t-1, t+1}^2}/{\bV_{t, t-1, t+1}}$ in distribution, where $\bL_{t, u_1, u_2}=(u_2-u_1)^{-3/2}[(u_2-t)\{W(t)-W(u_1)\}-(t-u_1)\{W(u_2)-W(t)\}]$ and $\bV_{t, t-1, t+1}=\int_{t-1}^t\bL_{u, t-1, t}^2du+\int_t^{t+1}\bL_{u, t, t+1}^2du$. (ii) The TUNE method, using the test statistic $S_{\tau,\rm mean}$ and setting the threshold $t_{u,\alpha}=\inf\{t:\Pr(\sup_{1\leq s<\infty}{\bL_{s, s-1, s+1}^2}/{\bV_{s, s-1, s+1}}>t)\le\alpha\}$, satisfies FWER $\le\alpha+o(1)$.
\end{proposition}

\subsection{True Nulls Regarding Neighboring Changepoints}\label{sec:H0_3}

Until now, our focus has been on evaluating true nulls as defined in (\ref{H0}), which assesses the presence of a changepoint within a local neighborhood $\th_j\pm h$ for a predetermined window size $h$. This allows practitioners to specify the desired precision for detecting changepoint locations \citep{Jewell+Fearnhead+Witten-2022-p1082}. Our numerical results also affirm that the FWER can be effectively controlled across a broad range of $h$ values; see Section \ref{SM:h} of Supplemental Material. In the literature, an alternative definition of true nulls considers changes between neighboring changepoints \citep{Hyun+Lin+GSell+Tibshirani-2021-p1037,Jewell+Fearnhead+Witten-2022-p1082}, specifically,
\begin{align}\label{H0_3}
\begin{aligned}
\widetilde{H}_{0,\th_j}: \theta_{\th_{j-1}+1}=\cdots=\theta_{\th_{j+1}}\ \text{versus}\ \widetilde{H}_{1,\th_j}:\ \text{not}\ \widetilde{H}_{0,\th_j},\quad\text{for}\ j=1,\ldots,\Kh.
\end{aligned}
\end{align}
Under this definition, a rejection of $\widetilde{H}_{0,\th_j}$ might occur due to changes that are far away from $\th_j$. This is one of the reasons why \cite{Jewell+Fearnhead+Witten-2022-p1082} recommend employing (\ref{H0}) as opposed to (\ref{H0_3}). However, considering $\widetilde{H}_{0,\th_j}$ can sometimes be informative, particularly during initial screening stages where precise localization of changepoints is less critical. Our general principle of TUNE can be extended to such scenarios with slight modifications. 

In this context, consider the two-sample test statistic as $T_{\tau\in(\underline{\tau},\overline{\tau}]}=\T(\{Z_i\}_{i=\underline{\tau}+1}^{\tau},\{Z_i\}_{i=\tau+1}^{\overline{\tau}})$ for a generic operator $\T$, and report $\R=\{\th_j\in\D:T_{\th_j\in(\th_{j-1},\th_{j+1}]}>t_{\rm u,\alpha}\}$ as a subset of $\D$ deemed reliable. For instance, in Example \ref{ex:uni_mean} with $\sigma=1$, one can employ
\[
T_{\tau\in(\underline{\tau},\overline{\tau}],\rm mean}=\frac{\bar{Z}_{(\underline{\tau},\tau]}-\bar{Z}_{(\tau,\overline{\tau}]}}{\sqrt{(\tau-\underline{\tau})^{-1}+(\overline{\tau}-\tau)^{-1}}}.
\]
Define $\widetilde{\mathcal{H}}$ as the set of all potential segments containing no changepoints:
\[
\widetilde{\mathcal{H}}=\{(\underline{\tau},\overline{\tau}]\subset(0,n): (\underline{\tau},\overline{\tau})\cap\{\tau_k^*\}_{k=1}^{K^*} = \emptyset\}.
\]
Thus, $\widetilde{H}_{0,\th_j}$ holds if and only if $(\th_{j-1},\th_{j+1}]\in\widetilde{\mathcal{H}}$. The FWER is then defined by
\begin{align*}
{\rm FWER}\equiv\Pr(\text{there exists some $\th_j\in\R$ and $(\th_{j-1},\th_{j+1}]\in\widetilde{\mathcal{H}}$}).
\end{align*}
Proposition \ref{prop:H0_3} demonstrates that this approach provides valid FWER control, with nullifiable test statistics---extending Assumption \ref{asmp:null}---and properly selected thresholds.

\begin{proposition}\label{prop:H0_3}
If (i) (Threshold selection) $\Pr_0\big(\max_{0\le\underline{\tau}<\tau<\overline{\tau}\le n}T_{\tau\in(\underline{\tau},\overline{\tau}]}>t_{\rm{u},\alpha}\big)\leq\alpha+o(1)$ and (ii) (Nullifiability) $\Pr\big(\max_{{0\le\underline{\tau}<\tau<\overline{\tau}\le n;(\underline{\tau},\overline{\tau}]\in\widetilde{\mathcal{H}}}}T_{\tau\in(\underline{\tau},\overline{\tau}]}>t_{\rm u,\alpha}\big) = \Pr_0\big(\max_{{0\le\underline{\tau}<\tau<\overline{\tau}\le n;(\underline{\tau},\overline{\tau}]\in\widetilde{\mathcal{H}}}}T_{\tau\in(\underline{\tau},\overline{\tau}]}>t_{\rm u,\alpha}\big) + o(1)$, then FWER $\le\alpha+o(1)$.
\end{proposition}

In Example \ref{ex:uni_mean} with $N(0,1)$ noises, the statistic $T_{\tau\in(\underline{\tau},\overline{\tau}],\rm mean}$ immediately satisfies the nullifiability condition, and the threshold can be approximated via simulation. Such statistics resembles scanning statistics used for changepoint detection. For example, \cite{Fang+Li+Siegmund-2020-p1615} investigated approximations of tail probabilities of $\max_{0\le\underline{\tau}<\tau<\overline{\tau}\le n}T_{\tau\in(\underline{\tau},\overline{\tau}],\rm mean}$ under a null hypothesis that no changes occur.

\section{Simulation Studies}\label{sec:simul}

To evaluate the performance of the TUNE method for post-detection inference, we conduct simulation studies across various changepoint models and data configurations, employing diverse detection algorithms. These simulations involve scenarios with equally spaced changepoints, i.e., $\tau_k^*=kn/(K^*+1)$ for $k\in[K^*]$, targeting a nominal FWER of $\alpha=5\%$. Each scenario is replicated $200$ times to compute empirical FWER and power, with power defined as the expected ratio of the number of true non-nulls deemed reliable to the number of true non-nulls detected, that is, $\E(|\R\cap\mathcal{H}^c|/|\D\cap\mathcal{H}^c|)$, where $\mathcal{H}^c=\{\tau\in[n-1]:\tau\not\in\mathcal{H}\}$.

\subsection{Univariate Mean Change Models}

Consider Example \ref{ex:uni_mean}. In the detection phase, we employ two variants of the BS algorithm---$k$-step BS and BS with a threshold $\lambda$---alongside the PELT algorithm with a penalty $\gamma$. These tuning parameters---$k$, $\lambda$, and $\gamma$---are either fixed at theoretically optimal values assuming normality ($K^*$, $\sqrt{2\log n}$, $\log n$, respectively) or determined via BIC ($\widehat{k}_{\rm BIC}$, $\widehat{\lambda}_{\rm BIC}$, and $\widehat{\gamma}_{\rm BIC}$, respectively); refer to \cite{yao1988estimating} and \cite{Fryzlewicz-2014-p2243}.

For post-detection inference, we implement two selective methods: JFW \citep{Jewell+Fearnhead+Witten-2022-p1082} and CF \citep{Carrington+Fearnhead-2023}, originally designed for fixed tuning parameters settings in BS and PELT but here adaptably applied for BIC-tuned parameters. In cases of unknown error variance, a robust variance estimator, as suggested by \cite{Jewell+Fearnhead+Witten-2022-p1082}, is used. As a comparative baseline, we use a sample-splitting approach, deploying odd-indexed observations for detection and even-indexed ones for inference, although this approach may reduce detection accuracy compared to both the selective methods and our TUNE method, which utilize the full dataset. To adjust for multiplicity, both the sample-splitting and selective methods incorporate Bonferroni corrections.

\subsubsection{IID Normal Noises with Known Variance}

\textbf{Scenario I} investigates an idealized setup, i.e., Example \ref{ex:uni_mean} with $P_{\epsilon}=N(0,\sigma^2)$, assuming a known variance $\sigma^2=1$. In this scenario, we conduct experiments with $n=500$ and $K^*=4$. We fix $\theta_1^*=1$, and set $\theta_{k+1}^*-\theta_k^*=(-1)^{k+1}\Delta$ for $k\in[K^*]$. To ensure comparability among different inference methods, we specify $h=10$. A more detailed examination of the effect of varying $h$ is presented in Section \ref{SM:h} of Supplemental Material.

\begin{figure}[!h]
\centering
\includegraphics[width=\textwidth]{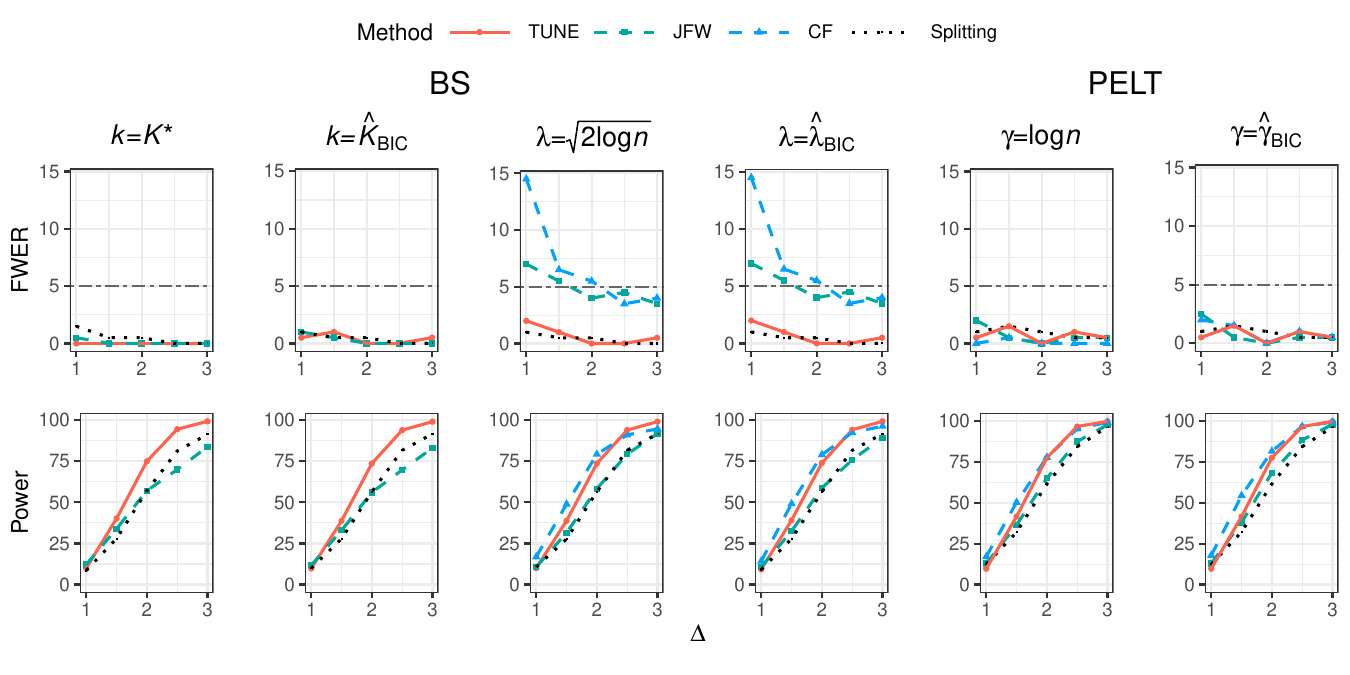}
\caption{Empirical FWER and power (in $\%$) under Scenario I for the TUNE, selective, and sample-splitting methods. The dot-dashed line represents the nominal FWER of $5\%$.}
\label{fig:I_FWER_power}
\end{figure}

Figure \ref{fig:I_FWER_power} showcases the empirical FWER and power across various detection schemes. The sample-splitting method consistently maintains FWER control across all configurations. However, it compromises the accuracy of changepoint detection (refer to Figure \ref{fig:I_hausdorff} in Supplemental Material) and power, due to the insufficient data usage for both detection and inference. On the other hand, the selective methods, JFW and CF, generally control FWER effectively through Bonferroni adjustments but exhibit increased FWERs when using the BS algorithm with BIC-selected thresholds, particularly in scenarios with weak signals. This increase is likely due to inflated selective errors when integrating model selection criteria; see Table \ref{tab:I_selective} in Supplemental Material. In addition, the CF method displays superior power relative to JFW, owing to less conditioning \citep{Carrington+Fearnhead-2023}. In contrast, our TUNE method (cf. Section \ref{sec:mean_uni_normal}) demonstrates robust FWER control across all settings, surpasses both the sample-splitting and the JFW method in power, and delivers comparable power to the CF method. The running times of all methods, detailed in Figure \ref{fig:I_time} in Supplemental Material, highlights that the TUNE method is more computational efficient relative to the selective approaches.

\subsubsection{Departures from Idealized Settings}

\textbf{Scenario II} explores deviations from the idealized conditions of Scenario I by introducing factors such as nonnormality, nonconstant variance, temporal correlations, or the presence of outliers. Specific sub-scenarios include:
\begin{itemize}
\item Scenario II(i) modifies Example \ref{ex:uni_mean} by adopting a t-distribution with $6$ degrees of freedom for $P_\epsilon$, standardized to ensure zero mean and unit variance.
\item Scenario II(ii) employs Example \ref{ex:Wald} with scores $\psi_\theta(Z_i)=Z_i/\theta-1$, corresponding to a Poisson distribution $P_i$ with parameter $\theta_i$.
\item Scenario II(iii) discards the assumption of independent noise in Example \ref{ex:uni_mean}. Instead, $\epsilon_i=0.5(1+\varepsilon_i)\epsilon_{i-1}+\varepsilon_i$, with $\varepsilon_i$ being iid from $N(0,1/2)$.
\item Scenario II(iv) revisits Example \ref{ex:uni_mean} but with  $P_\epsilon=0.8N(0,1)+0.2N(5,1)$.
\end{itemize}
For all sub-scenarios, experiments are conducted with $n=500$ and $K^*=4$, setting $\theta^*_k$ as in Scenario I, except in Scenario II(ii) where $\theta_{k+1}^*-\theta_k^*=(-1)^{k+1}\delta$ with $\Delta=\delta/\sqrt{(\delta+2)/2}$ to accommodate variance heterogeneity. We vary $\Delta$ to assess performance across various signal levels. We specify $h=20$.

For changepoint detection, the BS algorithm with threshold $\lambda=\sqrt{2\log n}$ is utilized, incorporating the detection statistic $T_{\tau,\rm mean}$ as in Scenario I but with $\sigma$ replaced by the robust estimator as used in JFW and CF. This detection setup may not be  ideally tailored for the considered sub-scenarios, potentially leading to inaccurate detection outcomes. Figure \ref{fig:II_FWER_power} depicts the empirical FWER and power for the selective methods, JFW and CF, alongside the TUNE method as in Section \ref{sec:mean_uni_normal} but replacing $\sigma^2$ by the aforementioned error variance estimator. All three methods encounter significant FWER inflation and apparently elevated power in Scenarios II(ii)--(iv), largely due to inappropriate selection of the detection and inference statistics.

\begin{figure}[!h]
\centering
\includegraphics[width=\textwidth]{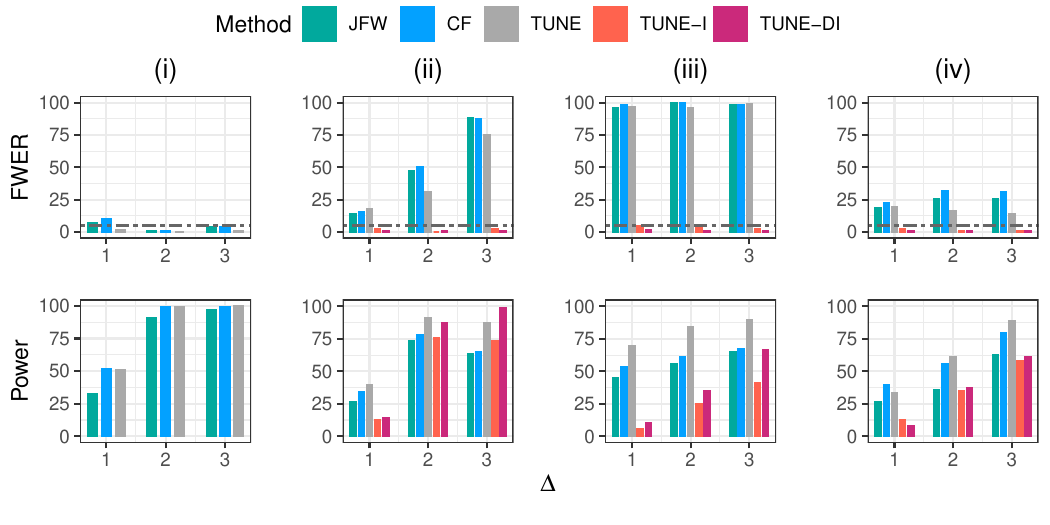}
\caption{Empirical FWER and power (in $\%$) under Scenario II for selective and TUNE methods.}
\label{fig:II_FWER_power}
\end{figure}

Adapting the JFW and CF approaches to flexible detection algorithms and inferential schemes poses great challenges. In contrast, the TUNE method offers flexible adaptability for each specific sub-scenario. We offer two strategies. \textit{Strategy I} (``I'' for modifying the inference statistic) implements customized inferential schemes, following guidelines from Section \ref{sec:Wald} for Scenarios II(ii), Section \ref{sec:self-normalization} for Scenario II(iii), and Section \ref{sec:nonpara} for Scenario II(iv), respectively. For Scenario II(ii), difference-based variance estimators $\widehat\sigma^2_{\tau}=\{2(2h-1)\}^{-1}\sum_{i=\tau-h}^{\tau+h-1}(Z_{i+1}-Z_i)^2$ are employed. With Strategy I, the TUNE-I method successfully maintains FWER control across Scenario II(ii)--(iv), delivering satisfactory power performance. \textit{Strategy DI} builds on Strategy I by further incorporating tailored changepoint detection procedures for each sub-scenario---``D'' for modifying the detection scheme. Specifically, for Scenario II(ii), the BS algorithm is implemented with a cross-validation-based criterion to determine the number of changepoints \citep{Zou+Wang+Li-2020-p413}, using the detection statistic $T_{\tau,\rm Wald}$. Scenario II(iii) utilizes a self-normalization-based detection algorithm to adjust for temporal correlations \citep{MR4515555}. For Scenario II(iv), a rank-based detection method \citep{MR3210993,MR3647098} is applied for handling data with outliers. With Strategy DI, the TUNE-DI method attains improved power due to more accurate detection outcomes, while maintaining robust FWER control, affirming the TUNE method's algorithm-agnostic robustness and reliability across diverse settings.

\subsubsection{Alternative True Null Settings}

We assess the effectiveness of the TUNE method as outlined in Section \ref{sec:H0_3}, particularly in handling post-detection inference related to true null hypotheses as in (\ref{H0_3}). We compare its performance against the selective methods, JFW and CF. The experiments are conducted as in Scenario I. The BS algorithm with threshold $\lambda=\sqrt{2\log n}$ is employed for changepoint detection. Figure \ref{fig:H0_3} presents the empirical FWER and power. All methods effectively control FWER in this scenario, while TUNE achieves higher power compared to the selective methods. The selective methods' loss of power is partly due to overconditioning in the calculation of selective p-values associated with such true nulls.

\begin{figure}[!h]
\centering
\includegraphics[width=.6\textwidth]{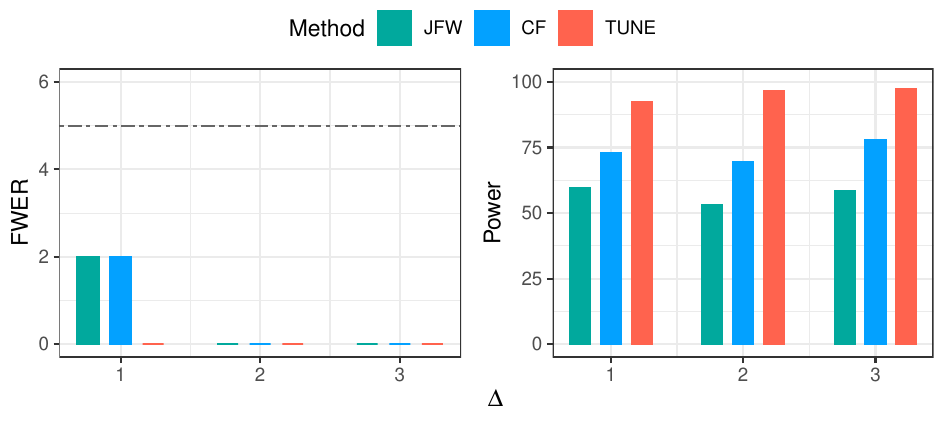}
\caption{Empirical FWER and power (in $\%$) regrading true nulls as in (\ref{H0_3}).}
\label{fig:H0_3}
\end{figure}

\subsection{Beyond Univariate Mean Changepoint Models}

This section delves into post-detection inference for more complex multivariate mean and regression changepoint models:
\begin{itemize}
\item Scenario III employs Example \ref{ex:mean} with $n=500$, $d\in\{5,200\}$, and $P_\epsilon=N(0,\Sigma)$, where $\Sigma=(0.5^{\vert i-j\vert})$.
\item Scenario IV focuses on a regression model as detailed in Example \ref{ex:score}, with $n=1000$, $d\in\{5,50\}$, $X_i$ iid from $N(0,\Sigma_X)$, and $P_\epsilon=N(0,1)$, where $\Sigma_X=(0.5^{\vert i-j\vert})$.
\item Scenario V extends Scenario III, introducing outliers by setting $P_{\epsilon}=\prod_{j=1}^dP_{\epsilon_j}$ with independent marginal distributions $P_{\epsilon_j}=0.9N(0,1)+0.1N(10,1)$, with $n=500$ and $d=5$.
\end{itemize}

For all scenarios, $K^*=4$. The mean and regression coefficient parameters are set such that for $j\le 5$, $\theta_{k,j}^*=1+(-1)^{k+1}\Delta/2$, and for $j>5$, $\theta_{k,j}^*=0$, for all $k\in[K^*+1]$. We vary $\Delta$ to explore different levels of signal. 

In the changpoint detection phase, for Scenario III, the inspect detection algorithm \citep{wang+Samworth+2018+p57} is used, for Scenario IV, the $2K^*$-step BS algorithm in conjunction with the Reliever device \citep{qian+wang+zou+2023} for accelerated computations is employed, and for Scenario V, the ecp algorithm \citep{MR3180567} is implemented.

\begin{figure}[!h]
\centering
\includegraphics[width=\textwidth]{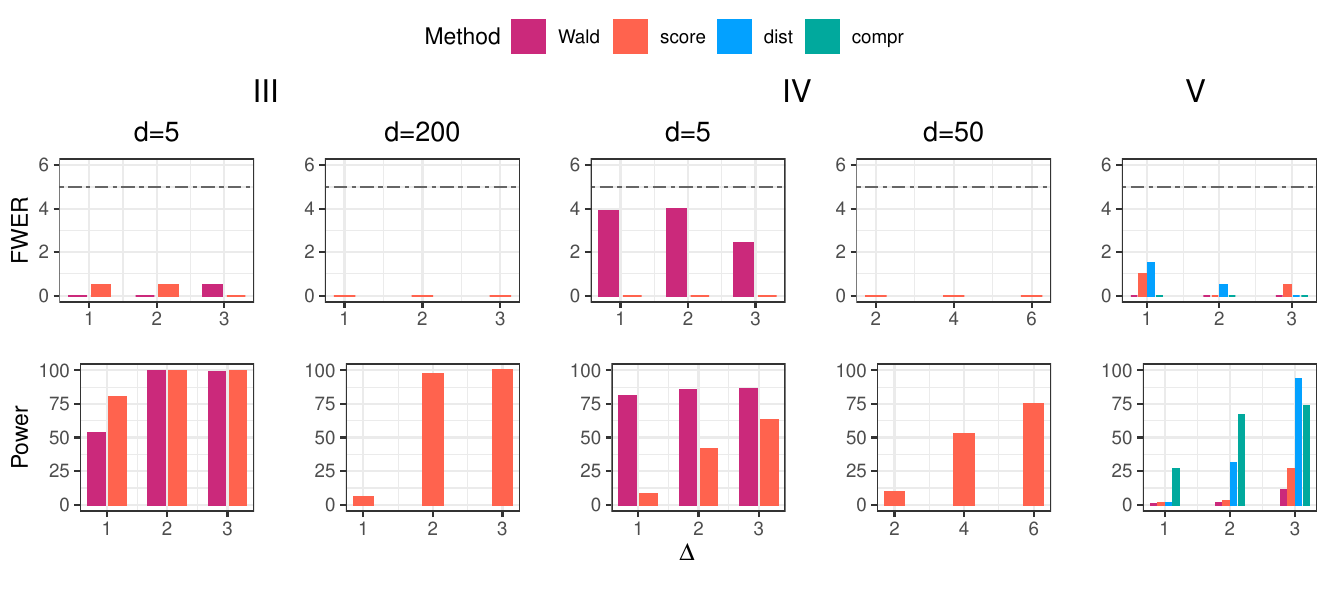}
\caption{Empirical FWER and power (in $\%$) under Scenarios III--V.}
\label{fig:others}
\end{figure}

In post-detection inference phase, the TUNE method is implemented using the score-based strategy across all scenarios (cf. Section \ref{sec:score}), selecting $g(\cdot)=\Vert\cdot\Vert_2$ for $d\le 5$ and $g(\cdot)=\Vert\cdot\Vert_\infty$ for $d\ge 50$, with $200$ bootstrap replications. For Scenarios III--IV with $d=5$, the Wald-based strategy is also applied (cf. Section \ref{sec:Wald}). Specifically, for Scenario III, $\psi_\theta(Z_i)=Z_i-\theta$ and $\widehat{\Gamma}_\tau=\{2(2n-1)\}^{-1}\sum_{i=1}^{n-1}(Z_{i+1}-Z_i)(Z_{i+1}-Z_i)^\top$, and for Scenarios IV, $\psi_\theta(Z_i)=-X_i(y_i-X_i^\top\theta)$ and $\widehat{\Gamma}_\tau=\widehat{\sigma}^{2}_{\tau}\widehat\Sigma^{-1}$, where $\widehat\Sigma$ is the sample covariance matrix based on $\{X_i\}_{i=1}^n$, and $\widehat{\sigma}^2_\tau=\{2(2h-1)\}^{-1}\sum_{i=\tau-h}^{\tau+h-1}(\widehat\epsilon_{i+1}-\widehat\epsilon_i)^2$ are based on differences of estimated residuals, $\widehat\epsilon_i=y_i-X_i^\top\widehat\theta_{(\tau-h,\tau+h]},i\in (\tau-h,\tau+h]$. Additionally, for Scenario V, the TUNE method is implemented using the statistics $T_{\tau,\rm compr}$ and $T_{\tau,\rm dist}$ (cf. Section \ref{sec:nonpara}). We specify $h=20$ for Scenarios III and V and $h=50$ for Scenario IV.

Figure \ref{fig:others} illustrates that different variants of the TUNE method provide robust FWER control across a variety of contexts, demonstrating their adaptability.

\section{Real Data Examples}\label{sec:realdata}

\subsection{GC-Content Data}\label{sec:GC}

This study analyzes GC (guanine-cytosine) content data from human chromosome $1$ to identify variations in nucleotide composition along the chromosome. The dataset, comprising $5000$ measurements, is accessed via the R package \texttt{changepoint}. For validation purposes, we hold out $2500$ even-indexed measurements, performing both changepoint detection and post-detection inference on the the remaining $n=2500$ odd-indexed measurements.

\begin{figure}[!h]
\centering
\includegraphics[width=0.8\textwidth]{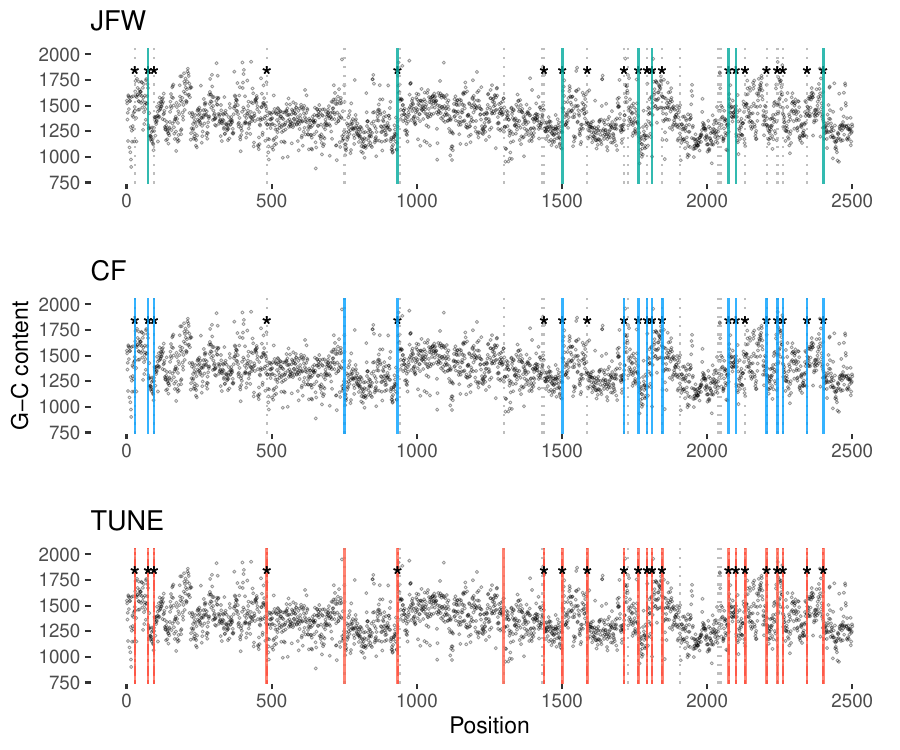}
\caption{Detected and validated changepoints in GC-Content Data.}
\label{fig:GC}
\end{figure}

Changepoint detection is implemented using the BS algorithm with threshold determined by BIC. The detected changepoints are indicated by dotted lines in each panel of Figure \ref{fig:GC}. To assess the reliability of these detected changepoints, the held-out validation dataset is employed, performing a sequence of local t-tests with a window size $h=20$, adjusted using Bonferroni corrections. The changepoints confirmed as reliable by this validation approach are also highlighted with asterisks in Figure \ref{fig:GC}.

Transitioning to the odd-indexed measurements for post-detection inference, we employ selective methods---JFW and CF---as well as our TUNE method (cf. Section \ref{sec:Wald}), maintaining $h=20$. The robust variance estimator, recommended by \cite{Jewell+Fearnhead+Witten-2022-p1082}, is used for all methods. Figure \ref{fig:GC} presents the changepoints validated by each method, marked with vertical lines. Out of the total $29$ detected changepoints, $8$ and $19$ are confirmed as reliable by the JFW and CF methods, respectively, which appear conservative. In contrast, TUNE validates $23$ changepoints, demonstrating greater sensitivity and power. Notably, of these $23$ validated changepoints by TUNE, $21$ coincide with those identified by the validation dataset, with only $2$ missed by the validation approach. This outcome affirms the TUNE method's enhanced diagnostic accuracy, while effectively controlling false positives.

\subsection{Array CGH Data}

The next study utilizes array CGH (comparative genomic hybridization) data to detect DNA sequence copy number variations in individuals with bladder tumors. The dataset, available from the R package \texttt{ecp}, including $\log$ intensity ratios fluorescent of DNA segments for $d=43$ individuals, measured across $2215$ loci. 

We follow the validation approach in Section \ref{sec:GC}, performing changepoint detection and TUNE-based inference on $n=1108$ odd-indexed observations. In the detection stage, we employ the inspect and ecp algorithms, along with the changeforest method---a classifier-based nonparametric detection approach \citep{londschien2023random}. These algorithms detect $312$, $28$, and $60$ changepoints, respectively, using their default settings. Following the recommendation in \cite{wang+Samworth+2018+p57}, we retain the most significant $30$ changepoints for both inspect and changeforest. The TUNE method, specifically tailored for this analysis, employs a score-based strategy with $g(\cdot)=\Vert\cdot\Vert_\infty$ and $B=200$ bootstrap replications, using a window size $h\in\{20,50\}$. The number of changepoints deemed reliable by the TUNE method for each algorithm is presented in Table \ref{tab:aGCH} (the column labeled ``TUNE''). 

\begin{table}[!h]
\begin{center}
\caption{Count of detected changepoints and validated changepoints in array CGH data.}
\begin{tabular}{lccccc}
\toprule
&&Detection&Validation&TUNE&Common \\
\midrule
\multirow{3}{*}{$h=20$}&inspect&30&19&10&10\\
&ecp&28&20&10&10\\
&changeforest&30&20&11&11\\\cline{2-6}
\multirow{3}{*}{$h=50$}&inspect&30&25&26&25\\
&ecp&28&23&24&23\\
&changeforest&30&26&27&25\\
\bottomrule
\end{tabular}
\label{tab:aGCH}
\end{center}
\end{table}

For validation, the $1107$ even-indexed observations are utilized. We conduct local high-dimensional two-sample mean tests, as proposed by \cite{Xue2020distribution}, to verify the reliability of the detected changepoints, using the window size $h\in\{20,50\}$. The number of changepoints validated by this test for each algorithm is also reported in Table \ref{tab:aGCH} (Column ``Validation''). Additionally, the table includes the count of changepoints identified by TUNE that coincide with those validated by the held-out dataset. These results indicate that TUNE effectively identifies reliable changepoints with minimal false positives and demonstrates enhanced power as the window size increases. Detailed visual summaries of these changepoints across several individuals are provided in Figure \ref{fig:aGCH} in Supplemental Material.

\section{Concluding Remarks}\label{sec:conclusion}

This paper introduces TUNE, a novel approach to post-detection inference in changepoint analysis, which has recently gained significant attention. TUNE is model- and algorithm-agnostic, offering robust control over the FWER irrespective of the detection algorithm employed and applicable to a diverse range of changepoint models. The core of TUNE involves a test statistic founded on the concept of nullifiability, alongside a universal threshold derived by approximating the distribution of the maximum of a sequence of localized two-sample statistics. We have outlined practical strategies for implementing this framework under various parametric changepoint models, including those in high-dimensional settings. TUNE provides a reliable and flexible alternative to selective methods for assessing the uncertainty of detected changepoints.

Our work opens several avenues for future research. Extending the principles of nullifiability and threshold determination to nonparametric changepoint models presents an exciting challenge, particularly when test statistics are developed by integrating advanced machine learning techniques \citep{li2024automatic}. Additionally, the foundational principles of TUNE could prove valuable in other areas of post-selection inference, such as post-clustering inference \citep{Gao2024cluster,chen2023selective}.

{\small
\baselineskip 10pt
\bibliographystyle{rss}
\bibliography{ref}
}

\newpage

\def\thesection{S.\arabic{section}}
\setcounter{section}{0}
\def\theequation{S.\arabic{equation}}
\setcounter{equation}{0}
\def\thelemma{S.\arabic{lemma}}
\setcounter{lemma}{0}
\def\thetable{S\arabic{table}}
\setcounter{table}{0}
\def\thefigure{S\arabic{figure}}
\setcounter{figure}{0}
\def\theassumption{S.\arabic{assumption}}
\setcounter{assumption}{0}
\def\thedefinition{S.\arabic{definition}}
\setcounter{definition}{0}

\begin{center}
\bf\Large Supplementary Material for\\
``TUNE: Algorithm-Agnostic Inference after Changepoint Detection''
\end{center}

The supplementary material includes proofs of Theorems \ref{thm:Wald}--\ref{thm:compr} and Propositions \ref{prop:power}--\ref{prop:H0_3}, and additional numerical results.

\section{Theoretical Proofs}\label{apdx:sec:proofs}

It should be noted that the positive constants $C$ and $c$ are employed throughout this section and may vary from line to line.

\subsection{Proof of Theorem \ref{thm:Wald}}

It suffices to verify that the threshold satisfies Eq. \eqref{thresh} and that the statistic $T_{\tau, \rm Wald}$ fulfills Assumption \ref{asmp:null}. The conclusion then follows directly from Theorem \ref{thm:FWER}.

Under Assumption \ref{asmp:Wald} and $\Pr_0$, Theorem 1(a) from \cite{Kirch2024} establishes that $a(n/h)T_{\tau, \rm Wald}-b(n/h)$ converges in distribution to the Gumbel extreme value distribution $G$. Given that $t_{\rm{u}, \alpha}= {b(n/h)+G^{-1}(1-\alpha)}/a(n/h)$, it follows that $\Pr(\max_{h\leq \tau\leq n-h}T_{\tau, \rm Wald}>t_{\rm{u}, \alpha})\leq \alpha+o(1)$, verifying Eq. \eqref{thresh}.

We proceed to show that $T_{\tau,\rm Wald}$ satisfies Assumption \ref{asmp:null} by controlling errors due to Gaussian approximations, linear expansions, and variance estimation, following arguments similar to those in \cite{Kirch2024}. For each $k\in[K^*+1]$ and $\tau\in\mathcal{H}_k=[\tau^*_{k-1}+h,\tau^*_k-h]$, define
\[
S^{*(k)}_\tau=(2h)^{-1/2}\big\|\Sigma^{*-1/2}_k\big\{\sum_{i\in(\tau-h,\tau]}\psi_{\theta^*_k}(Z_i)-\sum_{i\in(\tau,\tau+h]}\psi_{\theta^*_k}(Z_i)\big\}\big\|_2.
\]
Under Assumptions \ref{asmp:Wald}(ii) and \ref{asmp:Wald}(iv), by the invariance principle for partial sums of independent random vectors (cf. Lemma \ref{apdx:lem:invariance}), there exists a $d$-dimensional standard Wiener process $\{W(t):t\ge 0\}$ such that, 
\[
\max_{k\in[K^*+1]}\max_{\tau\in\mathcal{H}_k}S^{*(k)}_\tau = \max_{\tau\in\mathcal{H}}(2h)^{-1/2}\|\{W(\tau)-W(\tau-h)\}-\{W(\tau+h)-W(\tau)\}\|_2+o_P(a^{-1}(n/h)),
\]
where $\mathcal{H}=\cup_{k\in[K^*+1]}\mathcal{H}_k$. Next, define
\[
W^{*(k)}_\tau = \sqrt{h/2}\|\Gamma_k^{*-1/2}(\widehat{\theta}_{(\tau-h,\tau]}-\widehat{\theta}_{(\tau,\tau+h]})\|_2.
\]
Assumption \ref{asmp:Wald}(iii) ensures that
\[
\max_{k\in[K^*+1]}\max_{\tau\in\mathcal{H}_k}W^{*(k)}_\tau=\max_{k\in[K^*+1]}\max_{\tau\in\mathcal{H}_k}S^{*(k)}_\tau+o_P(a^{-1}(n/h)).
\]
Moreover, Assumption \ref{asmp:Wald}(v) guarantees
\[
\max_{k\in[K^*+1]}\max_{\tau\in\mathcal{H}_k}T_{\tau,\rm Wald} = \max_{k\in[K^*+1]}\max_{\tau\in\mathcal{H}_k}W^{*(k)}_\tau + o_P(a^{-1}(n/h)).
\]
Combining the above results, we obtain
\begin{align*}
\max_{\tau\in\mathcal{H}}T_{\tau,\rm Wald}
&= \max_{k\in[K^*+1]}\max_{\tau\in\mathcal{H}_k}T_{\tau,\rm Wald}\\
&= \max_{\tau\in\mathcal{H}}(2h)^{-1/2}\|\{W(\tau)-W(\tau-h)\}-\{W(\tau+h)-W(\tau)\}\|_2 + o_P(a^{-1}(n/h)).
\end{align*}
Repeating the same arguments under $\Pr_0$, we similarly have
\begin{align*}
\max_{\tau\in\mathcal{H}}T_{\tau,\rm Wald}
&= \max_{\tau\in\mathcal{H}}(2h)^{-1/2}\|\{W(\tau)-W(\tau-h)\}-\{W(\tau+h)-W(\tau)\}\|_2 + o_P(a^{-1}(n/h)).
\end{align*}
Since the leading terms in both expressions are identical under $\Pr$  and $\Pr_0$, it follows that $\Pr(\max_{\tau\in\mathcal{H}}T_{\tau, \rm Wald}>t_{\rm{u}, \alpha})=\Pr_0(\max_{\tau\in\mathcal{H}}T_{\tau, \rm Wald}>t_{\rm{u}, \alpha})+o(1)$, confirming Assumption \ref{asmp:null}. 

Therefore, Theorem \ref{thm:Wald} holds.

\subsection{Proof of Theorem \ref{thm:mean}}

Theorem \ref{thm:mean} is a direct consequence of Theorem \ref{thm:score}.

\subsection{Proof of Theorem \ref{thm:score}}

It suffices to show that
\begin{align*}
\Pr_0^*&(\max_{\tau=h,\ldots,n-h}T_{\tau,\rm score}>t_{\rm u,\alpha})\\
&\le\alpha + (nhd)^{-c} +C\big(\{(B_n^2+\Delta_\mu^2)\log^7(nhd)/h\}^{1/6}+\{\Delta_\Sigma\log^2(nhd)/h\}^{1/3}\big),
\end{align*}
for some constants $c,C>0$.

Denote $\epsilon_i=\S_i-\E(\S_i)$. Let \(\bZ_i \sim N(0, \bSigma_i)\) be independent Gaussian random vectors with covariance matrices $\Sigma_i=\Var(\epsilon_i)$ for $i\in[n]$. Let $S_{\tau, \S}=(2h)^{-1/2}(\sum_{i=\tau-h+1}^\tau \S_i-\sum_{i=\tau+1}^{\tau+h}\S_i)$ be the employed statistics, \(S_{\tau, \bZ} = (2h)^{-1/2} ( \sum_{i=\tau-h+1}^{\tau} \bZ_i - \sum_{i=\tau+1}^{\tau+h} \bZ_i )\) be the Gaussian analogues, and $S_{\tau, \rm{boots}} = (4h)^{-1/2} \{ \sum_{i=\tau-h+1}^{\tau} \e_i (\S_{i+1} - \S_i) - \sum_{i=\tau+1}^{\tau+h} \e_i (\S_i - \S_{i-1}) \}$ be the bootstrap statistics. Notice that $T_{\tau,\rm score}=g(S_{\tau, \S})$ and $T_{\tau, \rm boots}=g(S_{\tau, \rm boots})$. Define $\eta_n=\max_{h\leq \tau_1, \tau_2\leq n-h}\Vert\Cov(S_{\tau_1, \bZ}, S_{\tau_2, \bZ})-\Cov(S_{\tau_1, \text{boots}}, S_{\tau_2,\text{boots}}\mid \{Z_i\}_{i=1}^n)\Vert_\infty$.

To proceed, we establish three key lemmas.

\begin{lemma}\label{lem:BEB-stat}
Under Assumption \ref{asmp:boots} with a fixed $s$, there exists a constant $C>0$ such that
\begin{align*}
\sup_{t\geq 0}\Big\lvert\Pr_0^*\big(\max_{h\leq\tau\leq n-h}T_{\tau,\rm score}\le t\big)-\Pr\big\{\max_{h\leq \tau\leq n-h}g\left(S_{\tau,\bZ}\right)\leq t\big\}\Big\rvert\leq C\{B_n^2\log^7(nhd)/h\}^{1/6}.
\end{align*}
\end{lemma}

\begin{lemma}\label{apdx:lem:multi-boot-diff}
Under Assumption \ref{asmp:boots} with a fixed $s$, for any sequence $\bar{\eta}_n>0$, on the event $\{\eta_n\leq \bar{\eta}_n\}$, we have
\begin{align*}
\sup_{t\geq 0}&\Big\vert\Pr\big\{\max_{h\leq \tau\leq n-h}T_{\tau, \rm boots}\leq t\mid\{Z_i\}_{i=1}^n\big\}-\Pr\big\{\max_{h\leq \tau\leq n-h}g\left(S_{\tau,\bZ}\right)\leq t\big\}\Big\vert\\
&\leq C\{\bar{\eta}_n^{1/3}\log^{2/3}(nhd)+\log^{1/2}(nhd)/h\},
\end{align*}
for some constant $C>0$.
\end{lemma}

\begin{lemma}\label{apdx:prop:boots_cov}
Under Assumption \ref{asmp:boots}(ii), there exist constants $c,C>0$ such that $\Pr(\eta_n>\bar{\eta}_n)\leq (nhd)^{-c}/2$, where $\bar{\eta}_n=C[\{(B_n^2+\Delta_\mu^2)\log^3(nhd)/h\}^{1/6}+(\Delta_\Sigma/h)^{1/3}]^3$.
\end{lemma}

Define $t_{\text{u}, \alpha}^{\bZ}=\inf\{t: \Pr\{\max_{h\leq \tau\leq n-h}g(S_{\tau,\bZ})> t\}\leq \alpha\}$. Notice that
\begin{align*}
\Pr_0^*&\big\{\max_{h\leq \tau\leq n-h}T_{\tau,\rm score}>t_{\text{u}, \alpha}\big\}\\
&\leq \Pr\big\{\max_{h\leq\tau\leq n-h}g(S_{\tau, \bZ})>t_{\text{u}, \alpha}^{\bZ}\big\}+\Big\vert\Pr_0^*\big\{\max_{h\leq \tau\leq n-h}T_{\tau,\rm score}>t_{\text{u}, \alpha}\big\}-\Pr_0^*\big\{\max_{h\leq \tau\leq n-h}T_{\tau,\rm score}>t_{\text{u}, \alpha}^{\bZ}\big\}\Big\vert\\
&~~~~~~+\Big\vert\Pr_0^*\big\{\max_{h\leq \tau\leq n-h}T_{\tau,\rm score}>t_{\text{u}, \alpha}^{\bZ}\big\}-\Pr\big\{\max_{h\leq \tau\leq n-h}g(S_{\tau, \bZ})>t_{\text{u}, \alpha}^{\bZ}\big\}\Big\vert.
\end{align*}
By applying Lemmas \ref{lem:BEB-stat}--\ref{apdx:prop:boots_cov} and following the proof of Theorem 3.6 in \cite{chen2018gaussian}, we have
\begin{align*}
\Pr_0^*&\big\{\max_{h\leq \tau\leq n-h}T_{\tau,\rm score}>t_{\text{u}, \alpha}\big\}\\
&\leq \alpha+2\sup_{t\geq 0}\Big\vert\Pr\big\{\max_{h\leq \tau\leq n-h}T_{\tau, \rm boots}\leq t\mid\{Z_i\}_{i=1}^n\big\}-\Pr\big\{\max_{h\leq \tau\leq n-h}g\left(S_{\tau,\bZ}\right)\leq t\big\}\Big\vert\\
&~~~~~~+2\Pr(\eta_n>\bar{\eta}_n)+3\sup_{t\geq 0}\Big\lvert\Pr_0^*\big(\max_{h\leq\tau\leq n-h}T_{\tau,\rm score}\le t\big)-\Pr\big\{\max_{h\leq \tau\leq n-h}g\left(S_{\tau,\bZ}\right)\leq t\big\}\Big\rvert\\
&\leq \alpha+C\{\bar{\eta}_n^{1/3}\log^{2/3}(nhd)+\log^{1/2}(nhd)/h\}+(nhd)^{-c}+C\{B_n^2\log^7(nhd)/h\}^{1/6}\\
&\leq \alpha+(nhd)^{-c}+ C\big(\{(B_n^2+\Delta_\mu^2)\log^7(nhd)/h\}^{1/6}+\{\Delta_\Sigma\log^2(nhd)/h\}^{1/3}\big).
\end{align*}
Therefore, Theorem \ref{thm:score} is proved.

\begin{proof}[Proof of Lemma \ref{lem:BEB-stat}]
Under $\Pr_0^*$, $S_{\tau, \S}=(2h)^{-1/2}\{\sum_{i=\tau-h+1}^{\tau}\epsilon_i-\sum_{i=\tau+1}^{\tau+h}\epsilon_i\}$ for $h\leq \tau\leq n-h$. Define the weights $\bomega_{i\tau}=(2h)^{-1/2}\bbI_{\{\tau-h<i\leq \tau\}}-(2h)^{-1/2}\bbI_{\{\tau<i\leq \tau+h\}}$. Then, $S_{\tau, \S}$ and its Gaussian analogue $S_{\tau,\bZ}$ can be written as weighted sums: $S_{\tau, \S}=\sum_{i=1}^n\bomega_{i\tau}\epsilon_i$ and $S_{\tau,\bZ}=\sum_{i=1}^n\bomega_{i\tau}\bZ_i$, respectively. Stack the $n-2h+1$ $d$-dimensional vectors $S_{\tau, \S}$ into a single $(n-2h+1)d$-dimensional vector: $\bfS_{\S}=(S_{h, \S}^{\top}, \dots, S_{n-h, \S}^{\top})^\top$. Similarly, define $\bfS_\bZ=(S_{h,\bZ}^{\top}, \dots, S_{n-h, \bZ}^{\top})^\top$. Let $\bfepsilon_i=(\bomega_{ih}\epsilon_i^\top,\ldots, \bomega_{i,n-h}\epsilon_i^\top)^\top$ and $\bfY_i=(\bomega_{ih}\bZ_i^\top,\ldots,\bomega_{i, n-h}\bZ_i^\top)^\top$. Then, $\bfS_{\S}=\sum_{i=1}^n\bfepsilon_i$ and $\bfS_\bZ=\sum_{i=1}^n\bfY_i$.

Define the function $\tg:\bbR^{(n-2h+1)d}\rightarrow\bbR$ by $\tg(\bfx)=\max_{1\leq k\leq n-2h+1}g(\bfx_{J_k})$, where $J_k=\{d(k-1)+1, \dots, dk\}$. By Assumption \ref{asmp:boots}(i), the set $\{x\in\bbR^{(n-2h+1)d}:\tg(x)\leq t\}=\cap_{k=1}^{n-2h+1}\{x\in\bbR^{(n-2h+1)d}: g(x_{J_{k}})\leq t\}$ is an intersection of $n-2h+1$ $s$-sparsely convex sets. By Definition \ref{def:s-sparsity}, such an intersection is also $s$-sparsely convex.

We aim to apply a Gaussian approximation theorem for high-dimensional random vectors over $s$-sparsely convex sets (cf. Lemma \ref{apdx:lem:gaussian:approx}). To do so, we need to verify Assumption \ref{apdx:asmp:gaussian-approx} based on Assumption \ref{asmp:boots}. For any $\bv\in\mathbb{S}^{d-1}$ with $\Vert \bv\Vert_0\leq s$, noticing that $\Var(S_{\tau, \S})=(2h)^{-1}\sum_{i=\tau-h+1}^{\tau+h}\bSigma_i$, we have $\E\{(v^\top \bfS_{\S})^2\}\geq b$. Given that  $\sum_{i=1}^{n}\vert\bomega_{i\tau}\vert^{2+k}=(2h)^{-k/2}$, we obtain $n^{-1}\sum_{i=1}^n\E\vert n^{1/2}\bfepsilon_{ij}\vert^{2+k}=n^{-1}\sum_{i=1}^n\E\vert n^{1/2}\bomega_{i\tau}\epsilon_{i\ell}\vert^{2+k}\leq B_n'^k$ for $j=(\tau-h)d+\ell$ and $\ell\in[d]$, where $B_n'=(2h)^{-1/2}n^{1/2}B_n$. Additionally, $\E[\exp\{\lvert n^{1/2}\bfepsilon_{ij}\rvert/B_n'\}]=\E[\exp\{\lvert n^{1/2}\bomega_{i\tau}\epsilon_{i\ell}\rvert/B_n'\}]\leq 2$. Thus, Assumption \ref{apdx:asmp:gaussian-approx} is satisfied. By Lemma \ref{apdx:lem:gaussian:approx}, we have
\begin{align*}
\sup_{t\geq 0}&\Big\lvert\Pr_0\big(\max_{h\leq\tau\leq n-h}T_{\tau,\rm score}\le t\big)-\Pr\big\{\max_{h\leq \tau\leq n-h}g\left(S_{\tau,\bZ}\right)\leq t\big\}\Big\rvert\\
&=\sup_{t\geq 0}\left\lvert\Pr_0^*\left\{\tg\left(\bfS_{\S}\right)\leq t\right\}-\Pr\left\{\tg\left(\bfS_{\bZ}\right)\leq t\right\}\right\rvert
\leq C\left\{B_n'^2\log^7(nhd)/n\right\}^{1/6}
\leq C\left\{B_n^2\log^7(nhd)/h\right\}^{1/6},
\end{align*}
completing the proof.
\end{proof}

\begin{proof}[Proof of Lemma \ref{apdx:lem:multi-boot-diff}]
Let $\bfS_{\text{boots}} = (S_{h, \text{boots}}^{\top}, \dots, S_{n-h, \text{boots}}^{\top})^\top$. Notice that $\max_{h\leq \tau\leq n-h}T_{\tau, \rm boots}=\tg(\bfS_{\rm{boots}})$ and $\eta_n=\|\Var(\bfS_{\bZ})-\Var(\bfS_{\rm{boots}}\mid \{Z_i\}_{i=1}^n)\|_\infty$. Lemma \ref{apdx:lem:multi-boot-diff} directly follows from Theorem 4.1 in \cite{chern2017central}.
\end{proof}

\begin{proof}[Proof of Lemma \ref{apdx:prop:boots_cov}]
Without loss of generality, assume that $(B_n^2+\Delta_\mu^2)\log^7(nhd)\leq h$. For $h\leq \tau_1\leq\tau_2\leq n-h$, notice that $\Cov(S_{\tau_1, \text{boots}}, S_{\tau_2, \text{boots}}\mid \{Z_i\}_{i=1}^n)$ and $\Cov(S_{\tau_1, \bZ}, S_{\tau_2, \bZ})$ are nonzero only when $\tau_1-h\leq \tau_2-h\leq \tau_1\leq \tau_2\leq \tau_1+h\leq \tau_2+h$ or $\tau_1-h\leq \tau_1\leq \tau_2-h\leq \tau_1+h\leq \tau_2\leq \tau_2+h$. We have $\eta_n\leq I+II+III+IV$, where 
\begin{align*}
I&=\max_{\tau_2-h\leq\tau_1\leq \tau_2\leq \tau_1+h, 1\leq j, k\leq d}\left\vert\frac{1}{4h}\sum_{i=\tau_2-h+1}^{\tau_1}\left\{\left(\S_{i+1, j}-\S_{i, j}\right)\left(\S_{i+1, k}-\S_{i, k}\right)-2\Sigma_{i, jk}\right\}\right\vert,\\
II&=\max_{\tau_2-h\leq\tau_1\leq \tau_2\leq \tau_1+h, 1\leq j, k\leq d}\left\vert\frac{1}{4h}\sum_{i=\tau_1+1}^{\tau_2}\left\{\left(\S_{i+1, j}-\S_{i, j}\right)\left(\S_{i, k}-\S_{i-1, k}\right)-2\Sigma_{i, jk}\right\}\right\vert,\\
III&=\max_{\tau_2-h\leq\tau_1\leq \tau_2\leq \tau_1+h, 1\leq j, k\leq d}\left\vert\frac{1}{4h}\sum_{i=\tau_2+1}^{\tau_1+h}\left\{\left(\S_{i, j}-\S_{i-1, j}\right)\left(\S_{i, k}-\S_{i-1, k}\right)-2\Sigma_{i, jk}\right\}\right\vert, \\
IV&=\max_{h\leq\tau_1\leq \tau_2-h\leq \tau_1+h\leq \tau_2\leq n-h, 1\leq j, k\leq d}\left\vert\frac{1}{4h}\sum_{i=\tau_2-h+1}^{\tau_1+h}\left\{\left(\S_{i+1, j}-\S_{i, j}\right)\left(\S_{i, k}-\S_{i-1, k}\right)-2\Sigma_{i, jk}\right\}\right\vert.
\end{align*}

We will demonstrate that each term is sufficiently small with high probability. Let us focus on term $I$, which we decompose further into six components: $I\leq C\sum_{i=1}^6I_i$, where
\begin{align*}
I_1&=\max_{\tau_2-h\leq\tau_1\leq \tau_2\leq \tau_1+h, 1\leq j, k\leq d}\left\vert\frac{1}{2h}\sum_{i=\tau_2-h+1}^{\tau_1}\left\{\epsilon_{i, j}\epsilon_{i, k}-\Sigma_{i, jk}\right\}\right\vert, \\
I_2&=\max_{\tau_2-h\leq\tau_1\leq \tau_2\leq \tau_1+h, 1\leq j, k\leq d}\left\vert\frac{1}{2h}\sum_{i=\tau_2-h+1}^{\tau_1}\left\{\epsilon_{i+1, j}\epsilon_{i+1, k}-\Sigma_{i+1, jk}\right\}\right\vert, \\
I_3&=\max_{\tau_2-h\leq\tau_1\leq \tau_2\leq \tau_1+h, 1\leq j, k\leq d}\left\vert\frac{1}{2h}\sum_{i=\tau_2-h+1}^{\tau_1}\epsilon_{i+1, j}\epsilon_{i, k}\right\vert, \\
I_4&=\max_{\tau_2-h\leq\tau_1\leq \tau_2\leq \tau_1+h, 1\leq j, k\leq d}\left\vert\frac{1}{2h}\sum_{i=\tau_2-h+1}^{\tau_1}\left(\mu_{i+1, j}-\mu_{i, j}\right)\left(\epsilon_{i+1, k}-\epsilon_{i, k}\right)\right\vert,\\
I_5&=\max_{h\leq \tau\leq n-h, 1\leq j\leq d}\left\vert\frac{1}{2h}\sum_{i=\tau-h+1}^{\tau+h-1}\left(\mu_{i+1, j}-\mu_{i, j}\right)^2\right\vert, \\
I_6&=\max_{h\leq \tau\leq n-h, 1\leq j, k\leq d}\frac{1}{2h}\sum_{i=\tau-h+1}^{\tau+h-1}\left\vert\Sigma_{i+1, jk}-\Sigma_{i, jk}\right\vert.
\end{align*}

\textit{Bounding $I_1$:} Following the proof of Proposition 4.1 in \cite{chern2017central}, for any $\zeta\in((nhd)^{-1}, e^{-1})$ and sufficiently large constant $C>0$, we have 
\begin{equation*}
\begin{aligned}
\Pr\left(I_1>C\{B_n^2\log^3(nhd)/h\}^{1/2}\right)\leq \Pr\left(I_1>C\{B_n^2\log(nhd)\log^2(1/\zeta)/h\}^{1/2}\right)\leq \zeta/20.
\end{aligned}
\end{equation*}

\textit{Bounding $I_2$ and $I_3$:} Similar arguments and bounds as for $I_1$ apply to $I_2$ and $I_3$.

\textit{Bounding $I_4$:} We decompose $I_4$ into two parts: $I_4\leq I_{4, 1}+I_{4, 2}$, where 
\begin{align*}
I_{4, 1}&=\max_{\tau_2-h\leq\tau_1\leq \tau_2\leq \tau_1+h, 1\leq j, k\leq d}\left\vert\frac{1}{2h}\sum_{i=\tau_2-h+1}^{\tau_1}\left(\mu_{i+1, j}-\mu_{i, j}\right)\epsilon_{i, k}\right\vert, \\
I_{4, 2}&=\max_{\tau_2-h\leq\tau_1\leq \tau_2\leq \tau_1+h, 1\leq j, k\leq d}\left\vert\frac{1}{2h}\sum_{i=\tau_2-h+1}^{\tau_1}\left(\mu_{i+1, j}-\mu_{i, j}\right)\epsilon_{i+1, k}\right\vert.
\end{align*}
To bounding $I_{4, 1}$, define 
\begin{align*}
U_{4, 1}&=\max_{1\leq i\leq n-1, 1\leq j, k\leq d}\left\vert\left(\mu_{i+1, j}-\mu_{i, j}\right)\epsilon_{i, k}\right\vert, \\
\phi_{4, 1}^2&=\max_{\tau_2-h\leq \tau_1\leq \tau_2\leq \tau_1+h, 1\leq j, k\leq d}\sum_{i=\tau_2-h+1}^{\tau_1}\E\left\{\left(\mu_{i+1, j}-\mu_{i, j}\right)^2\epsilon_{i, k}^2\right\}.
\end{align*}
Under Assumption \ref{asmp:boots}(ii) and noting that the number of change points is fixed, we have
\begin{align*}
\Vert U_{4, 1}\Vert_2&\leq \Vert U_{4, 1}\Vert_{\psi_{1/2}}\leq \Delta_\mu+\big\Vert\max_{1\leq i\leq n, 1\leq k\leq d}\epsilon_{i, k}^2\big\Vert_{\psi_{1/2}}\leq \Delta_\mu + CB_n^2\log^2\left(nhd\right), \\
\phi_{4, 1}^2&\leq \max_{\tau_2-h\leq \tau_1\leq \tau_2\leq \tau_1+h, 1\leq k\leq d}\sum_{i=\tau_2-h+1}^{\tau_1}\E\left(\Delta_\mu^2+\epsilon_{i, k}^4\right)\leq 2h(\Delta_\mu^2+ B_n^2),
\end{align*}
where $\Vert X\Vert_{\psi_{1/2}}=\inf\{t>0: \E\{\exp(\sqrt{\vert X\vert/t})\}\leq 2\}$. Applying Lemma \ref{apdx:lem:maximal-expec}, we obtain $\E(I_{4, 1})\leq C\{(\Delta_\mu^2+B_n^2)\log(nhd)/h\}^{1/2}$. Using Lemma \ref{apdx:lem:maximal-prob}, for any $t>0$,
\begin{equation*}
\begin{aligned}
\Pr&\left(I_{4, 1}>C\{(\Delta_\mu^2+B_n^2)\log(nhd)/h\}^{1/2}+t\right)\\
&\leq \exp\Big\{-\frac{2ht^2}{3(\Delta_\mu^2+B_n^2)}\Big\}+3\exp\Big\{-c\Big(\frac{2ht}{\Delta_\mu+B_n^2\log^2(nhd)}\Big)^{1/2}\Big\}.
\end{aligned}
\end{equation*}
Choose $t=C\{(\Delta_\mu^2+B_n^2)\log(nhd)/h\}^{1/2}\log(1/\zeta)$ with  $\zeta\in((nhd)^{-1}, e^{-1})$ and sufficient large constant $C>0$. Since $(B_n^2+\Delta_\mu^2)\log^7(nhd)\leq h$, we have
\begin{equation*}
\begin{aligned}
\Pr\left(I_{4, 1}>C\{(B_n^2+\Delta_\mu^2)\log^3(nhd)/h\}^{1/2}\right)\leq \zeta/20.
\end{aligned}
\end{equation*}
A similar bound holds for $I_{4, 2}$.

\textit{Bounding $I_5$ and $I_6$:} we have $I_5\leq\Delta_\mu/h$ and $I_6\leq\Delta_\Sigma/h$.

Combining the bounds on $I_1$ through $I_6$ and choosing an appropriate value of $\zeta$ (e.g., $\zeta=(nhd)^{-c}/2$), we find $\Pr(I>\bar{\eta}_n/4)\leq \zeta/4$.

\textit{Bounding $II$, $III$, and $IV$:} Similar techniques and arguments used for $I$ can be applied to $II$, $III$, and $IV$, yielding bounds of the same order.

Finally, summing the probabilities, we conclude that $\Pr(\eta_n>\bar{\eta}_n)\leq (nhd)^{-c}/2$.
\end{proof}

\subsection{Proof of Theorem \ref{thm:rank}}

Let $\mathcal{H}_k=(\tau^*_{k-1}+h,\tau^*_k-h]$ for $k\in[K^*+1]$ and $R^k_i$ be the rank of $Z_i$ among $Z_{\mathcal{H}_k}$, i.e., $R^k_i=\sum_{j\in\mathcal{H}_k}\bbI_{\{Z_j\le Z_i\}}$. Notice that $Z_j\le Z_i$ if and only if $R^k_j\le R^k_i$ for $i,j\in\mathcal{H}_k$. As a result, for $\tau\in\mathcal{H}_k$, 
\[
T_{\tau,\text{WMW}}
=\sum_{i=\tau+1}^{\tau+h}R_{\tau,i}
=\sum_{i=\tau+1}^{\tau+h}\sum_{j=\tau-h+1}^{\tau+h}\bbI_{\{R^k_j\le R^k_i\}},
\]
a function of ranks $R^k_i$. Consequently, $T_{\tau,\text{WMW}}$ is ($P^*_k$-)distribution-free. The independence among \{$Z_{\mathcal{H}_k}\}_{k\in[K^*+1]}$ ensures that $\max_{\tau\in\mathcal{H}}T_{\tau,\text{WMW}}$ has identical distribution under both $\Pr$ and $\Pr_0$, fulfilling Assumption \ref{asmp:null}. According to Theorem \ref{thm:FWER}, the conclusion follows.

\subsection{Proof of Theorem \ref{thm:compr}}

The statistic $T_{\tau,\rm compr}$ can be expressed as the $L_2$-norm of a U-statistic:
\[
T_{\tau,\rm compr}=\|2^{-1/2}h^{-3/2}\widehat{\Sigma}_\tau^{-1/2}\sum_{i=\tau-h+1}^\tau\sum_{j=\tau+1}^{\tau+h}\varphi(Z_i, Z_j)\|_2,
\]
where $\varphi(Z_i, Z_j)$ is a $d$-dimensional random vector with components $\varphi(Z_{i, k}, Z_{j, k})=\bbI_{\{Z_{i, k}\leq Z_{j, k}\}}-\bbI_{\{Z_{j, k}\leq Z_{i, k}\}}$ for $k\in[d]$; see also \cite{lung2015homogeneity}. For $i\in[n]$, define $\varphi_i(y)=\int \varphi(x, y)dF_{i}(x)$ and $\widetilde{\varphi}_i(x)=\int\varphi(x, y)dF_{i}(y)$. Due to the continuity of $F_i$, $\varphi_i(y)=2F_i(y)-1$ and $\widetilde{\varphi}_i(x)=1-2F_i(x)$, and consequently, $\varphi_i(x)=-\widetilde{\varphi}_i(x)$.

\textbf{Part I:} We first show that $\Pr_0(\max_{h\leq \tau\leq n-h}T_{\tau,\rm compr}>t_{\rm{u}, \alpha})\leq\alpha+o(1)$. Let $T_{\tau, {\rm compr, proj}}$ be the $L_2$-norm of the Hoeffding projection of $T_{\tau,\rm compr}$:
\[
T_{\tau, {\rm compr, proj}} = \|(2h)^{-1/2} \Sigma_\tau^{-1/2} \{\sum_{i=\tau-h+1}^{\tau} \widetilde{\varphi}_1(Z_i) + \sum_{j=\tau+1}^{\tau+h} \varphi_1(Z_j)\}\|_2,
\]
where $\Sigma_\tau=4\Var\{F_\tau(Z_\tau)\}$. Notice that
\begin{equation*}
\begin{aligned}
\big|\max_{h\leq \tau\leq n-h}T_{\tau,\rm compr}-\max_{h\leq \tau\leq n-h}T_{\tau, {\rm compr, proj}}\big|\leq C\big\{\max_{h\leq \tau\leq n-h}I_{ \tau}+\max_{h\leq \tau\leq n-h}II_{ \tau}\big\},
\end{aligned}
\end{equation*}
where 
\begin{align*}
I_{\tau}&=\left\|\frac{1}{\sqrt{2h^3}}\widehat{\Sigma}_\tau^{-1/2}\sum_{i=\tau-h+1}^{\tau}\sum_{j=\tau+1}^{\tau+h}\varphi(Z_i, Z_j)-\frac{1}{\sqrt{2h^3}}\Sigma_\tau^{-1/2}\sum_{i=\tau-h+1}^{\tau}\sum_{j=\tau+1}^{\tau+h}\varphi(Z_i, Z_j)\right\|_2,\\
II_{\tau}&=\left\|\frac{1}{\sqrt{2h^3}}\Sigma_\tau^{-1/2}\sum_{i=\tau-h+1}^{\tau}\sum_{j=\tau+1}^{\tau+h}\varphi(Z_i, Z_j)-\frac{1}{\sqrt{2h}}\Sigma_\tau^{-1/2}\left\{\sum_{i=\tau-h+1}^\tau\widetilde{\varphi}_1(Z_i)+\sum_{j=\tau+1}^{\tau+h}\varphi_1(Z_j)\right\}\right\|_2.
\end{align*}

Since $\E\{|\varphi_1(Z_1)|^2\}<\infty$, applying the invariance principle (cf. Lemma \ref{apdx:lem:invariance}), there exists a $d$-dimensional standard Wiener process $\{W(t):t\geq 0\}$ such that, 
\begin{equation}
\begin{aligned}\label{apdx:eq:proj}
\max_{h\leq \tau\leq n-h}&T_{\tau, {\rm compr, proj}}\\
&=\max_{1\leq t\leq n/h-1}2^{-1/2}\left\|\{W(t)-W(t-1)\}-\{W(t+1)-W(t)\}\right\|_2+o_P(\{\log(n/h)\}^{-1/2}).
\end{aligned}
\end{equation}
Using Lemma \ref{apdx:lem:wiener}(ii), we conclude that  $\max_{h\leq \tau\leq n-h}T_{\tau, {\rm compr, proj}}=O_P(\sqrt{\log(n/h)})$.

Next, we bound $\max_{h\leq \tau\leq n-h}II_{\tau}$. Let $\varrho(Z_i, Z_j)=\varphi(Z_i, Z_j)-h^{-1}\{\sum_{i=\tau-h+1}^{\tau}\widetilde{\varphi}_1(Z_i)+\sum_{j=\tau+1}^{\tau+h}\varphi_1(Z_j)\}$. Because $\Sigma_\tau$ is positive definite,
\begin{align*}
II_{\tau}
=\|(2h^3)^{-1/2}\Sigma_\tau^{-1/2}\sum_{i=\tau-h+1}^\tau\sum_{j=\tau+1}^{\tau+h}\varrho(Z_i, Z_j)\|_2
\leq C\max_{1\leq k\leq d}\left|\frac{1}{\sqrt{2h^3}}\sum_{i=\tau-h+1}^\tau\sum_{j=\tau+1}^{\tau+h}\varrho(Z_{i, k}, Z_{j, k})\right|.
\end{align*}
Notice that each term for $k\in[d]$ inside the absolute value is a degenerate U-statistic. Applying a Bernstein-type inequality for degenerate U-statistics (see Proposition 2.3(d) in \cite{Arcones1993limit}), for any $t>0$,
\begin{equation*}
\begin{aligned}
\Pr_0&\Big(\log^{1/2} (n/h)\max_{\substack{h\leq \tau\leq n-h,\\1\leq k\leq d}}\left|\frac{1}{\sqrt{2h^3}}\sum_{i=\tau-h+1}^\tau\sum_{j=\tau+1}^{\tau+h}\varrho(Z_{i, k}, Z_{j, k})\right|\geq t\Big)\\
&\leq nd \max_{\substack{h\leq \tau\leq n-h,\\1\leq k\leq d}}\Pr_0\Big(\left|\frac{1}{\sqrt{2h^3}}\sum_{i=\tau-h+1}^\tau\sum_{j=\tau+1}^{\tau+h}\varrho(Z_{i, k}, Z_{j, k})\right|\geq \log^{-1/2}(n/h)t\Big)\\
&\leq nd\exp\{-C\log^{-1}(n/h)h^2t^2\}.
\end{aligned}
\end{equation*}
Given that $\log^3 n/h\to 0$, it follows that $\max_{h\leq \tau\leq n-h}II_{\tau}=o_P(\{\log(n/h)\}^{-1/2})$.

We now proceed to bound $\max_{h\leq \tau\leq n-h}I_{\tau}$. From the above results, we have
\[
\max_{h\leq \tau\leq n-h}\|(2h^3)^{-1/2}\sum_{i=\tau-h+1}^\tau\sum_{j=\tau+1}^{\tau+h}\varphi(Z_i, Z_j)\|_2=O_P(\sqrt{\log(n/h)}).
\]
Notice that
\begin{align*}
\max_{h\leq \tau\leq n-h}\|\widehat{\Sigma}_\tau-\Sigma_\tau\|_2
&\leq C\max_{\substack{h\leq \tau\leq n-h,\\1\leq k\leq d}}\frac{1}{2h}\sum_{i=\tau-h+1}^{\tau+h}\{\widehat{F}_\tau(Z_{i, k})-F_\tau(Z_{i, k})\}^2\\
&~~+C\max_{\substack{h\leq \tau\leq n-h,\\1\leq k\leq d}}\frac{1}{2h}\sum_{i=\tau-h+1}^{\tau+h}\vert\widehat{F}_\tau(Z_{i, k})-F_\tau(Z_{i, k})\vert+\max_{h\leq \tau\leq n-h}\|\widehat{\Sigma}_{\tau, \rm{proj}}-\Sigma_\tau\|_2,  
\end{align*}
where $\widehat{\Sigma}_{\tau, \rm{proj}}=2h^{-1}\sum_{i=\tau-h+1}^{\tau+h}\{F_\tau(Z_i)-1/2\}\{F_\tau(Z_i)-1/2\}^\top$. Using the Dvoretzky–Kiefer–Wolfowitz inequality, for any $t>0$, 
\begin{align*}
\Pr_0&\Big(\log(n/h)\max_{\substack{h\leq \tau\leq n-h,\\1\leq k\leq d}}\frac{1}{2h}\sum_{i=\tau-h+1}^{\tau+h}\{\widehat{F}_\tau(Z_{i, k})-F_\tau(Z_{i, k})\}^2>t\Big)\\
&\leq Cn\max_{\substack{h\leq \tau\leq n-h,\\\tau-h<i\leq \tau+h,\\1\leq k\leq d}}\Pr_0\left(|\widehat{F}_\tau(Z_{i, k})-F_\tau(Z_{i, k})|>\log^{-1/2}(n/h)t^{1/2}\right)\\
&\leq Cn\exp\{-ch\log^{-1}(n/h)t\},\\
\Pr_0&\Big(\log(n/h)\max_{\substack{h\leq \tau\leq n-h,\\1\leq k\leq d}}\frac{1}{2h}\sum_{i=\tau-h+1}^{\tau+h}\vert\widehat{F}_\tau(Z_{i, k})-F_\tau(Z_{i, k})\vert>t\Big)\\
&\leq Cn\exp\{-ch\log^{-2}(n/h)t^2\}.
\end{align*}
Applying Lemma \ref{apdx:lem:covariance}, for any $t>0$, we obtain
\begin{align*}
\Pr_0&\Big(\log(n/h)\max_{h\leq \tau\leq n-h}\|\widehat{\Sigma}_{\tau, \rm{proj}}-\Sigma_\tau\|_
2\geq t\Big)\\
&\leq n\max_{h\leq \tau\leq n-h}\Pr_0\left(\|\widehat{\Sigma}_{\tau, \rm{proj}}-\Sigma_\tau\|_2\geq \log^{-1} (n/h)t\right)\\
&\leq 2nd\exp\left(-\frac{Ch\log^{-2}(n/h)t^2}{\max_{h\leq \tau\leq n-h}\|\Sigma_\tau\|_2+\log^{-1}(n/h)t}\right).
\end{align*}
Since $\max_{h\leq \tau\leq n-h}\|\Sigma_\tau\|_2<\infty$ and $\log^3n/h\to0$, it follows that $\max_{h\leq \tau\leq n-h}\|\widehat{\Sigma}_\tau-\Sigma_\tau\|_2=o_P(\{\log(n/h)\}^{-1})$. 
Combining these results, we conclude that $\max_{h\leq \tau\leq n-h}I_{\tau}=o_P(\log^{-1/2}(n/h))$.

Therefore, we have
\begin{align}\label{apdx:eq:remainder}
\big|\max_{h\leq \tau\leq n-h}T_{\tau,\rm compr}-\max_{h\leq \tau\leq n-h}T_{\tau, {\rm compr, proj}}\big|=o_P(\{\log(n/h)\}^{-1/2}).
\end{align}
Combining Eq. \eqref{apdx:eq:proj} and Eq. \eqref{apdx:eq:remainder}, we obtain
\[
\max_{h\leq \tau\leq n-h}T_{\tau,\rm compr}=\sup_{1\leq t<\infty}2^{-1/2}\|\{W(t)-W(t-1)\}-\{W(t+1)-W(t)\}\|_2+o_P(\{\log(n/h)\}^{-1/2}).
\]
By the choice of $t_{\rm u,\alpha}$, $\Pr_0(\max_{h\leq \tau\leq n-h}T_{\tau,\rm compr}>t_{\rm{u}, \alpha})\leq \alpha+o(1)$.

\textbf{Part II:} Next, we verify Assumption \ref{asmp:null}. Consider $T_{\tau,\rm compr}$ for $\tau \in \mathcal{H}$. Under $\Pr_0$, following similar arguments, we have
\begin{align*}
\max_{\tau\in\mathcal{H}}T_{\tau,\rm compr}=\max_{\tau\in\mathcal{H}}(2h)^{-1/2}\|\{W(\tau)-W(\tau-h)\}-\{W(\tau+h)-W(\tau)\}\|_2+o_P(\{\log(n/h)\}^{-1/2}).
\end{align*}
In the presence of changepoints, for $k\in[K^*+1]$ and $i\in[\tau_{k-1}^*+1, \tau_k^*]$, $\varphi_i(x)=\varphi_{\tau_k^*}(x)$,  $\widetilde{\varphi}_i(y)=\widetilde{\varphi}_{\tau_k^*}(y)$, and $\Sigma_i=\Sigma_{\tau_k^*}$. Thus, for $k\in[K^*+1]$ and $\tau\in\mathcal{H}_k=[\tau_{k-1}^*+h, \tau_k^*-h]$, we can similarly show that
\begin{align*}
\max_{\tau\in\mathcal{H}}T_{\tau,\rm compr}=\max_{\tau\in\mathcal{H}}(2h)^{-1/2}\|\{W(\tau)-W(\tau-h)\}-\{W(\tau+h)-W(\tau)\}\|_2+o_P(\{\log(n/h)\}^{-1/2}).
\end{align*}
Thus, $\Pr(\max_{\tau\in\mathcal{H}}T_{\tau, \rm compr}>t_{\rm{u}, \alpha})=\Pr_0(\max_{\tau\in\mathcal{H}}T_{\tau, \rm compr}>t_{\rm{u}, \alpha})+o(1)$, verifying Assumption \ref{asmp:null}. 

The conclusion follows from Theorem \ref{thm:FWER}.

\subsection{Proof of Proposition \ref{prop:power}}

Observe that
\begin{align*}
M_{\tau,\rm mean} &\ge \sqrt{\frac{\{\tau^*_k-(\tau-h)\}(\tau+h-\tau^*_k)}{2h}}\left\vert\bar{Z}_{(\tau-h,\tau^*_k]}-\bar{Z}_{(\tau^*_k,\tau+h]}\right\vert\\
&= \sqrt{\frac{\{\tau^*_k-(\tau-h)\}(\tau+h-\tau^*_k)}{2h}}\left\vert\bar{\epsilon}_{(\tau-h,\tau^*_k]}-\bar{\epsilon}_{(\tau^*_k,\tau+h]}+(\theta^*_{k-1}-\theta^*_k)\right\vert\\
&\gtrsim \sqrt{\frac{m_n(2h-m_n)}{2h}}|\theta^*_{k-1}-\theta^*_k| - M_\epsilon.
\end{align*}
Using arguments similar to those in the proof of Theorem \ref{thm:Wald}, it can be concluded that $M_\epsilon=O_P(\sqrt{\log(n/h)})$, provided that Assumption \ref{asmp:power} holds. Hence, the conclusion follows.

\subsection{Proof of Proposition \ref{prop:sn}}

It suffices to show Proposition \ref{prop:sn}(i). Since $\sum_{i=1}^n\epsilon_i-\sigma W(n)=O(n^{1/(2+\nu)})$ almost surely (a.s.), under $\Pr_0$, we have $L_{\tau, \tau-h, \tau+h}=\sigma\bL_{t, t-1, t+1}+O(n^{1/(2+\nu)}/\sqrt{h})$, a.s., for $h\leq \tau\leq n-h$ and $t=\tau/h$. Applying Lemma \ref{apdx:lem:wiener}(i), $\max_{1\leq t\leq n/h-1}\bL_{t, t-1, t+1}=O(\sqrt{\log (n/h)})$ a.s.. Since $n/h\to\infty$ and $\{n^{2/(2+\nu)}\log n\}/h\to0$,  $L_{\tau, \tau-h, \tau+h}^2=\sigma^2 \bL_{t, t-1, t+1}^2+o(1)$, a.s.. Similarly, we can show that $L^2_{j, \tau-h, \tau}=\sigma^2\bL_{u, t-1, t}^2+o(1)$ and $L^2_{j, \tau, \tau+h}=\sigma^2\bL_{u, t, t+1}^2+o(1)$ a.s. with $u=j/h$. By continuous mapping theorem, $h^{-1}(\sum_{j=\tau-h+1}^\tau L_{j, \tau-h, \tau}^2+\sum_{j=\tau+1}^{\tau+h}L_{j, \tau, \tau+h}^2)=\bV_{t, t-1, t+1}+o(1)$ a.s.. Hence, we conclude that $\max_{h\leq \tau\leq n-h}S_{\tau, {\rm mean}}\to\sup_{1\leq t<\infty}\bL^2_{t, t-1, t+1}/\bV_{t, t-1, t+1}$ in distribution.

\subsection{Proof of Proposition \ref{prop:H0_3}}

\begin{align*}
\text{FWER} &= \Pr(T_{\th_j\in(\th_{j-1}, \th_{j+1}]}>t_{\text{u}, \alpha}\text{ for some }(\th_{j-1}, \th_{j+1}]\in\widetilde{\mathcal{H}})\\
&\stackrel{(i)}{\leq} \Pr\left(\max_{0\leq \underline{\tau}<\tau<\overline{\tau}\leq n ;(\underline{\tau}, \overline{\tau}]\in\widetilde{\mathcal{H}}}T_{\tau\in(\underline{\tau}, \overline{\tau}]}>t_{\text{u}, \alpha}\right)\stackrel{(ii)}{=}\Pr_0\left(\max_{0\leq \underline{\tau}<\tau<\overline{\tau}\leq n ;(\underline{\tau}, \overline{\tau}]\in\widetilde{\mathcal{H}}}T_{\tau\in(\underline{\tau}, \overline{\tau}]}>t_{\text{u}, \alpha}\right)+o(1)\\
&\stackrel{(iii)}{\leq} \Pr_0\left(\max_{0\leq \underline{\tau}<\tau<\overline{\tau}\leq n}T_{\tau\in(\underline{\tau}, \overline{\tau}]}>t_{\text{u}, \alpha}\right)+o(1)\stackrel{(iv)}{\leq} \alpha+o(1).
\end{align*}
Inequality (i) extends detected true nulls to all potential true nulls. Inequality (ii) applies the assumption of nullifiability. Inequality (iii) further expands the true null pool to all possible changepoint locations. Finally, Inequality (iv) is justified by the threshold selection assumption.

\section{Auxiliary Definitions and Lemmas}\label{apdx:sec:details}

\begin{definition}[Sparsely convex sets]\label{def:s-sparsity}
For some integer $s>0$, we say that $A\subset \bbR^d$ is an $s$-sparsely convex set if there exists an integer $Q>0$ and convex sets $A_q\subset \bbR^d$, $q\in[Q]$, such that $A=\cap_{q=1}^QA_q$ and the indicator function of each $A_q$, $w\mapsto \bbI(w\in A_q)$, depends on at most $s$ elements of its argument $w=(w_1, \dots, w_d)$. 
\end{definition}

Let $X_1, \dots, X_n$ be independent centered random vectors in $\bbR^d$ with $d\geq 3$. Assume that $\E(X_{ij}^2)<\infty$ for $i\in[n]$ and $j\in[d]$. Define the normalized sum $S_X=n^{-1/2}\sum_{i=1}^n X_i$. Let $Y_1, \dots, Y_n$ be independent centered Gaussian random vectors in $\bbR^d$ such that each $Y_i$ has the same covariance matrix as $X_i$, that is, $Y_i\sim N(0, \E(X_iX_i^\top))$. Define the normalized sum for the Gaussian random vectors: $S_Y=n^{-1/2}\sum_{i=1}^n Y_i$. Fix an integer $s>0$, and let $\mathcal{A}^{\text{sp}}(s)$ denote the class of all $s$-sparsely convex  sets in $\bbR^d$. 
\begin{assumption}\label{apdx:asmp:gaussian-approx}
(i) $n^{-1}\sum_{i=1}^n\E\{(v^\top X_i)^2\}\geq b$ for all $v\in\S^{d-1}$ with $\Vert v\Vert_0\leq s$. (ii) $n^{-1}\sum_{i=1}^n\E(\vert X_{ij}\vert^{2+k})\leq B_n^k$ for all $j\in[d]$ and $k=1, 2$. (iii) $\E\{\exp(\vert X_{ij}\vert/B_n)\}\leq 2$ for all $i\in[n]$ and $j\in[d]$.
\end{assumption}
\begin{lemma}\label{apdx:lem:gaussian:approx}
Suppose that Assumption \ref{apdx:asmp:gaussian-approx} is satisfied. We have
\begin{align*}
\sup_{A\in\mathcal{A}^{\text{sp}}(s)}\left\vert \Pr\left(S_X\in A\right)-\Pr\left(S_Y\in A\right)\right\vert\leq C\left\{B_n^2\log^7(nd)/n\right\}^{1/6},
\end{align*}
where the positive constant $C$ depends only on $s$ and $b$.
\end{lemma}
\begin{proof}
See Proposition 3.2 in \cite{chern2017central}.
\end{proof}

\begin{lemma}\label{apdx:lem:maximal-expec}
Let $X_1, \ldots, X_n$ be independent centered random vectors in $\mathbb{R}^d$ with $d \geq 2$. Define $ V= \max_{j\in[d]} \left|\sum_{i=1}^n X_{ij}\right|$, $U = \max_{i\in[n], j\in[d]} |X_{ij}|$ and $\phi^2= \max_{j\in[d]} \sum_{i=1}^n \mathbb{E}(X_{ij}^2)$. Then  
\[
\E(V) \leq C \left( \phi \sqrt{\log d} + \sqrt{\E(U^2)}\log d \right),
\]
where $C$ is a universal constant.
\end{lemma}
\begin{proof}
See Lemma E.1 in \cite{chern2017central}.
\end{proof}
\begin{lemma}\label{apdx:lem:maximal-prob}
Consider the setting of Lemma \ref{apdx:lem:maximal-expec}.  For every $\gamma > 0$, $\beta \in (0, 1]$ and $t > 0$,
\[
\mathbb{P}\left\{ V \geq (1 + \gamma) \mathbb{E}(V) + t \right\} \leq \exp\left\{-\frac{t^2}{3 \phi^2}\right\} + 3 \exp\left\{-\left(\frac{t}{C \| U\|_{\psi_\beta}}\right)^\beta \right\},
\]
where $\Vert U\Vert_{\psi_{\beta}}=\inf\{t>0: \E[\exp\{(\vert U_1\vert/t)^\beta\}]\leq 2\}$ and $C$ is a constant depending only on $\gamma$, $\beta$.
\end{lemma}
\begin{proof}
See Lemma E.2(i) in \cite{chern2017central}.
\end{proof}

\begin{lemma}\label{apdx:lem:invariance}
Denote $\mathscr{H}$ as the set of all continuous functions $H:[0, \infty)\to [0, \infty)$ such that $t^{-2}H(t)$ is non-decreasing and $t^{-4+r}H(t)$ is non-increasing for some $r>0$. Let $\{X_i:i\in[n]\}$ be a sequence of independent random vectors with zero means and $\Var(X_i)=\sigma_i^2\Gamma$. Suppose that there exists a non-negative divergent sequence $\{a_i\}$ such that $\sum_{i=1}^\infty H(a_i)^{-1}\E\{H(|X_i|)\}<\infty$ for some $H\in\mathscr{H}$. Then one can construct a sequence $\{Y_i:i\in[n]\}$, where $Y_i$ follows the distribution $N(0, \sigma_i^2\Gamma)$.  The partial sums $S_n = \sum_{i=1}^n X_i$ and $T_n = \sum_{i=1}^n Y_i$ satisfy $S_n - T_n = o(a_n)$ almost surely.
\end{lemma}
\begin{proof}
See Theorem 2 in \cite{Einmahl1987invariance}.
\end{proof}

\begin{lemma}\label{apdx:lem:wiener}
\begin{itemize}
\item[(i)] Let $\{W(t):t\geq 0\}$ be a standard Wiener process, then it holds
\begin{align*}
\sup_{0\leq t\leq n-1}\sup_{0\leq u\leq 1}\left|W(t+u)-W(t)\right| = O(\sqrt{\log n}), ~~\text{almost surely.}
\end{align*}
\item[(ii)] Let $\{W(t):t\geq0\}$ be a $d$-dimensional standard Wiener process with identity matrix $I_d$ as a covariance matrix, then it holds
\begin{align*}
\sup_{0\leq u\leq n-1}\sup_{0\leq u\leq 1}\left\|W(t+u)-W(t)\right\|_2 = O(\sqrt{\log n}), ~~\text{almost surely.}
\end{align*}
\end{itemize} 
\end{lemma}
\begin{proof}
See Lemmas E.2.1 and E.2.2 in \cite{reckruhm2019estimating}.
\end{proof}

\begin{lemma}\label{apdx:lem:covariance}
Let $X_1, \dots, X_n$ be iid zero-mean random vectors with covariance matrix $\Sigma$ such that $\|X_j\|_2\leq \sqrt{b}$ almost surely. Then for all $t>0$, the sample covariance matrix $\widehat{\Sigma}=n^{-1}\sum_{i=1}^nX_iX_i^\top$ satisfies
\begin{align*}
\Pr\left(\|\widehat{\Sigma}-\Sigma\|_2\geq t\right)\leq 2d\exp\left\{-\frac{nt^2}{2b(\|\Sigma\|_2+t)}\right\}.
\end{align*}
\end{lemma}
\begin{proof}
See Corollary 6.3 in \cite{wainwright2019high}.
\end{proof}

\section{Additional Numerical Results}\label{sec:impact_h}

\subsection{Scenario I: Hausdorff Distances}

\begin{figure}[!h]
\centering
\includegraphics[width=.8\textwidth]{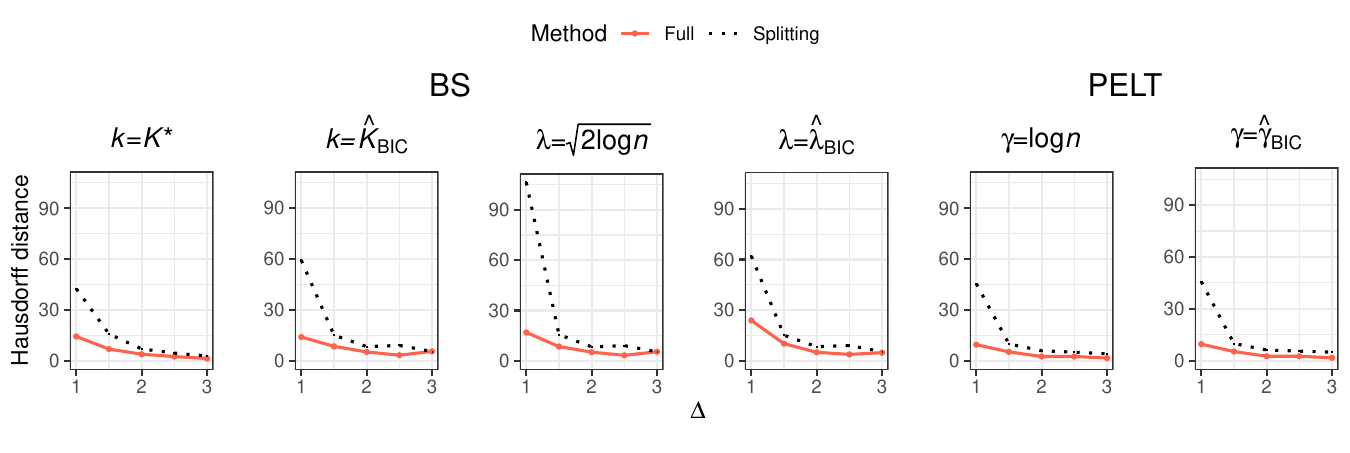}
\caption{Empirical Hausdorff distance between estimated and true changepoints under Scenario I, comparing results using the full dataset versus half the dataset (e.g., for the sample-splitting method).}
\label{fig:I_hausdorff}
\end{figure}

We assess changepoint detection accuracy using the Hausdorff distance between estimated and true changepoints
\[
\max\{\max_{1\leq i\leq K^*}\min_{1\leq j \leq K}\vert\tau_i^*-\widehat\tau_j\vert,\max_{1\leq j\leq K}\min_{1\leq i \leq K^*}\vert\tau_i^*-\widehat\tau_j\vert\}.
\]
It is evident that utilizing the full dataset results in a smaller Hausdorff distance, as illustrated in Figure \ref{fig:I_hausdorff}.

\subsection{Scenario I: Selective Errors}

For the selective methods, extra experiments are carried out with $\Delta=0$ (indicating no changes) to evaluate the selective error rate associated with a randomly selected true null from those detected, as explored in \cite{Jewell+Fearnhead+Witten-2022-p1082} and \cite{Carrington+Fearnhead-2023}, without applying Bonferroni corrections.

\begin{table}[!h]
\begin{center}
\caption{Empirical selective error (in $\%$) under Scenario I with $\Delta=0$. Entries marked with ``-'' indicate the method not available for those configurations.}
\begin{tabular}{lcccccc}
\toprule
\multirow{2}{*}{}&\multicolumn{4}{c}{BS}&\multicolumn{2}{c}{PELT}\\ \cmidrule(r){2-5}\cmidrule(r){6-7}
& $k=K^*$ &$k=\widehat{K}_{\rm BIC}$ & $\lambda=\sqrt{2\log n}$ & $\lambda=\widehat{\lambda}_{\rm BIC}$ & $\gamma=\log n$ & $\gamma=\widehat{\gamma}_{\rm BIC}$ \\
\midrule
JCW&3.5&30.5&0.5&15.0&5.0&23.55\\
CF&-&-&0.5&19.0&16.5&34.5\\
\bottomrule
\end{tabular}
\label{tab:I_selective}
\end{center}
\end{table}

\subsection{Scenario I: Computation Costs}

\begin{figure}[!h]
\centering
\includegraphics[width=.8\textwidth]{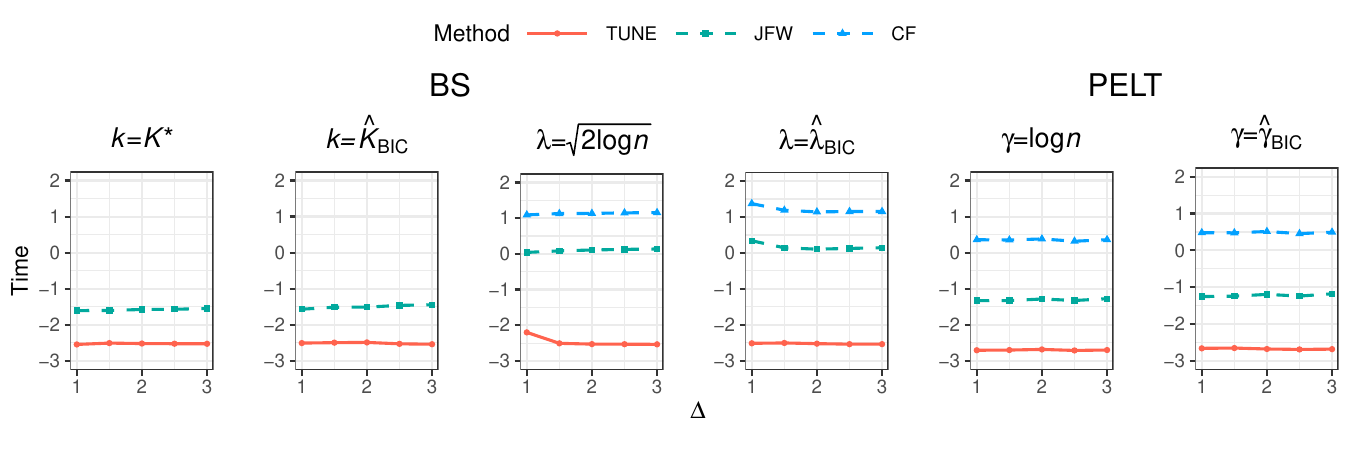}
\caption{Computation time (in seconds) under Scenario I for the TUNE and selective methods. The y-axis is displayed in $\log_{10}$-scale.}
\label{fig:I_time}
\end{figure}

Figure \ref{fig:I_time} displays the average running time across each detection scheme. The computation cost of the selective method, as proposed by \cite{Carrington+Fearnhead-2023}, increases with a parameter $N$, which is a critical component of this method. This increase is evident from the comparison between $N=1$ (reduced to the JFW method) and $N=10$ (the CF method considered in our paper). Conversely, our method exhibits negligible runtime while maintaining power comparable to that of the CF method, as illustrated in Figure \ref{fig:I_FWER_power}.

\subsection{Scenario I: Effect of Window Size}\label{SM:h}

\begin{figure}[!h]
\centering
\includegraphics[width=.8\textwidth]{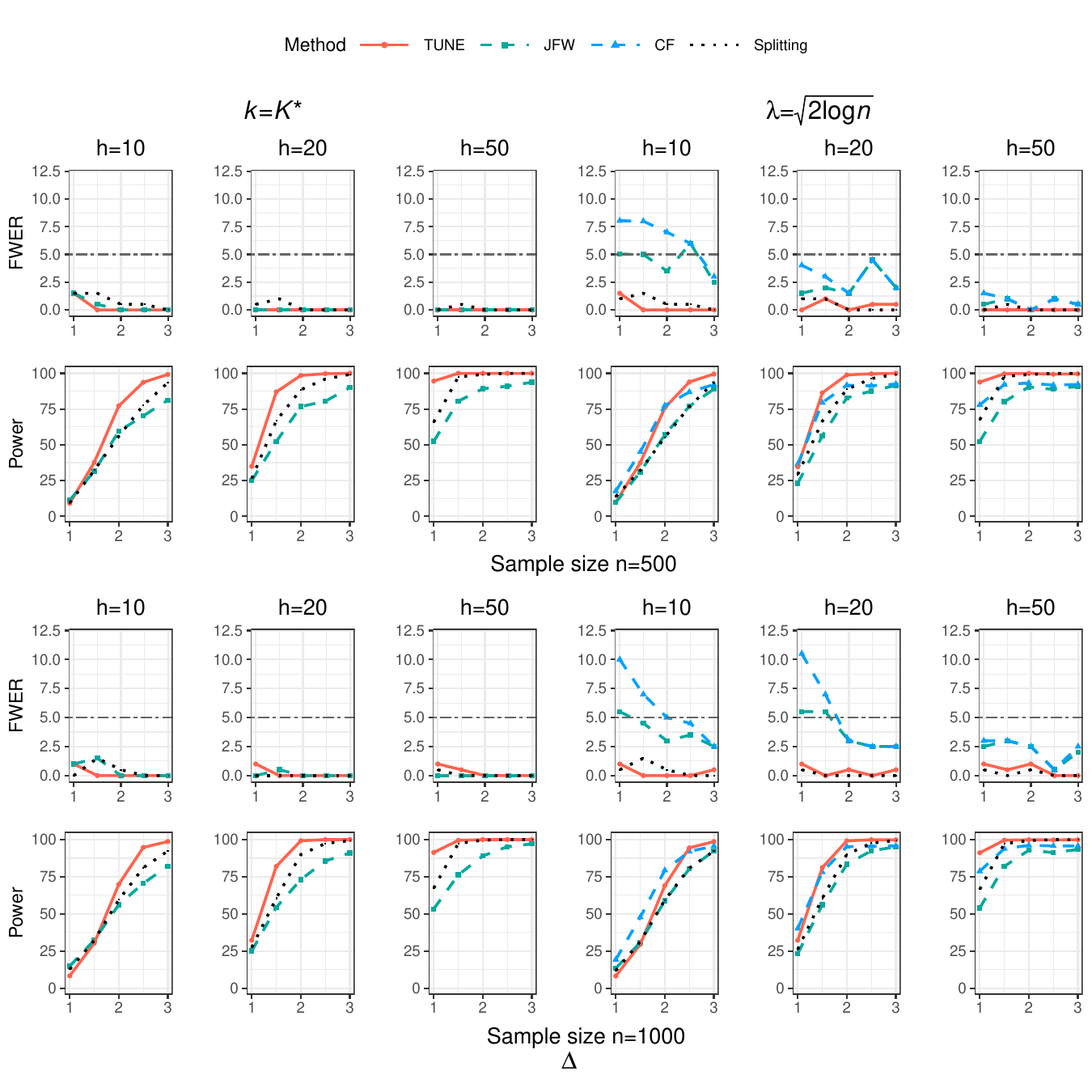}
\caption{Empirical FWER and power (in $\%$) for TUNE, selective and sample-splitting methods under Scenario I, with different sample size $n$, signal strength $\Delta$, and window size $h$.}
\label{fig:I_h}
\end{figure}

To assess the impact of varying the window size $h$, we revisit Scenario I, with $n\in\{500,1000\}$ and $h\in\{10,20,50\}$. Figure \ref{fig:I_h} illustrates the empirical FWER and power for the $K^*$-step BS and BS with threshold $\lambda=\sqrt{2\log n}$. The TUNE method consistently maintains robust FWER control across all window sizes, while the CF method experiences FWER inflation at smaller $h$ values sometimes. In terms of power, TUNE not only outperforms the sample-splitting and JFW methods but also achieves comparable power to the CF method at smaller $h$ values and surpasses it as $h$ increases. 

\subsection{Array CGH Data: Detected and Validated Changepoints}

Figure \ref{fig:aGCH} provides visualizations of detected and validated changepoints for window sizes $h=20$ and $50$. Changepoints detected by TUNE align more closely with those identified by the validation approach as $h$ increases, indicating enhanced power. Notably, the visual representation reveals that changepoints tend to cluster at specific loci, particularly when $h=20$ is used. This clustering suggests that these loci may be common sites of genomic instability or possess biological significance in the context of bladder tumors, meriting further investigation.

\begin{figure}[!h]
\centering
\subfigure[$h=20$]{
\includegraphics[width=0.5\textwidth]{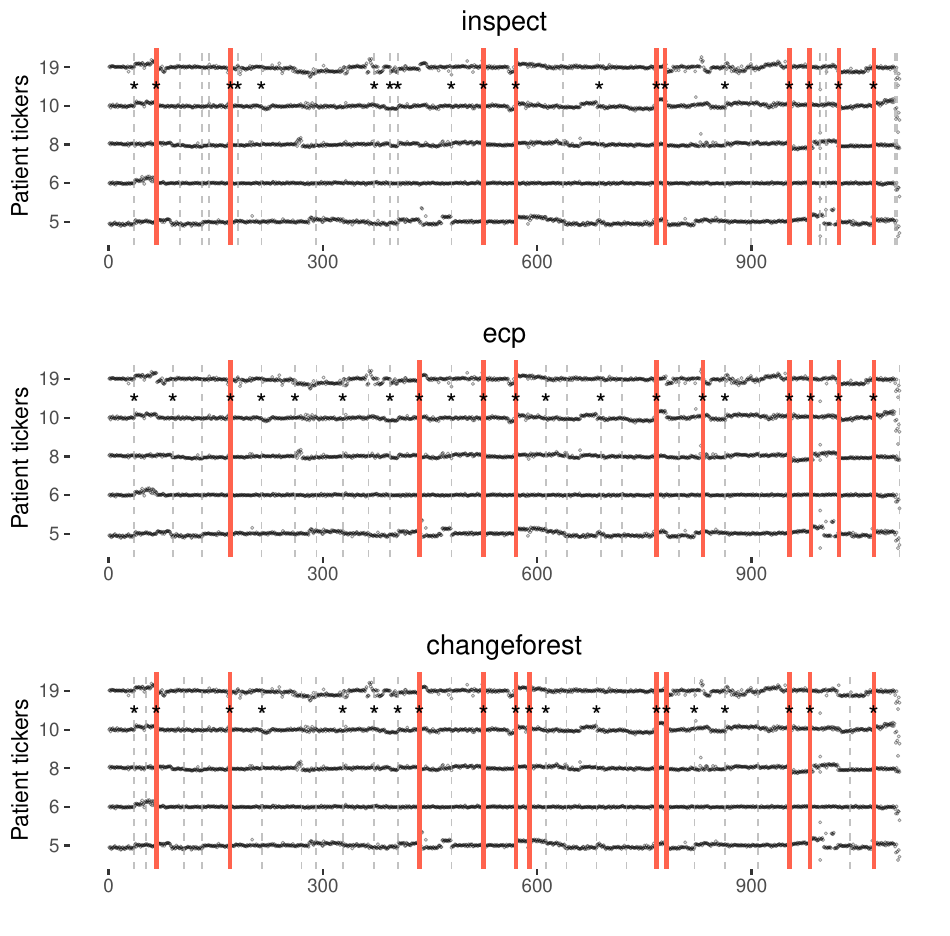}}\subfigure[$h=50$]{
\includegraphics[width=0.5\textwidth]{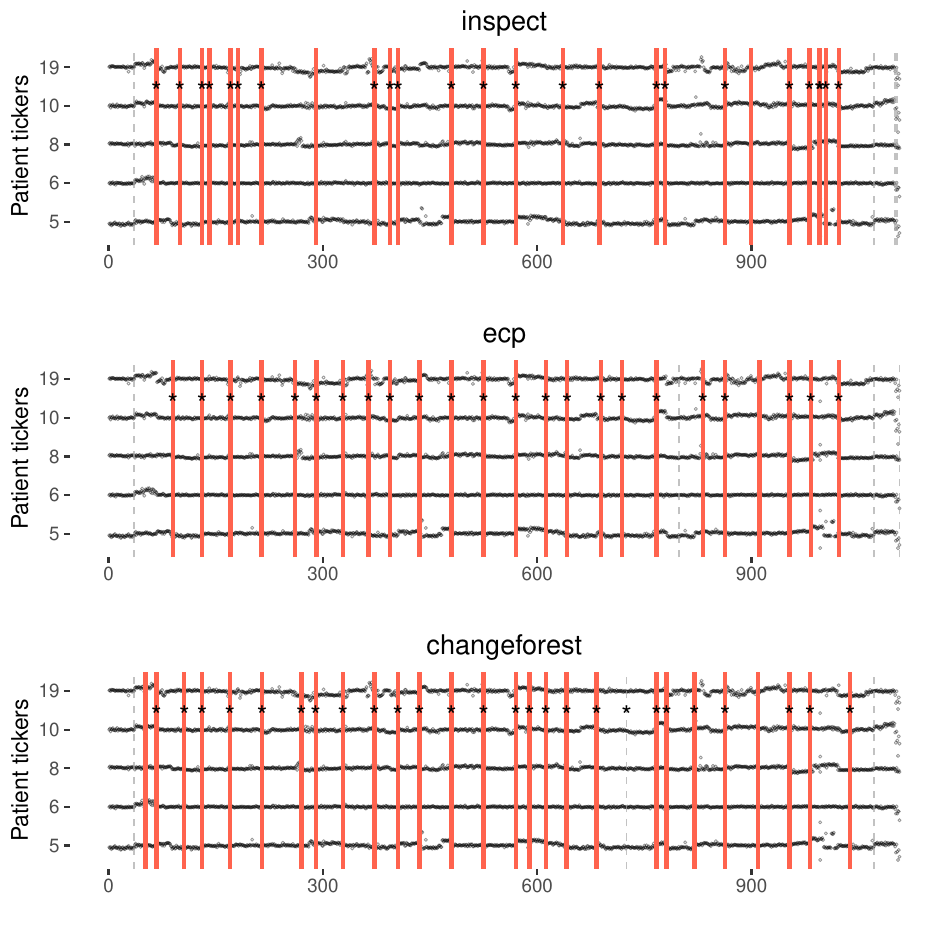}}
\caption{Detected and validated changepoints in the array CGH dataset for $h\in\{20,50\}$. Data points are represented by circles. Detected changepoints are indicated by dotted lines, and those validated by the TUNE method are indicated by vertical lines. Changepoints validated by the held-out dataset are marked with asterisks.}
\label{fig:aGCH}
\end{figure}

\end{document}